\documentclass[a4paper,UKenglish,cleveref,autoref,thm-restate]{lipics-v2021}

\pdfoutput=1 \hideLIPIcs  

\usepackage[utf8]{inputenc}
\usepackage{latexsym}
\usepackage{amsmath}
\usepackage{amsthm}
\usepackage{amssymb}
\usepackage{mathtools}
\usepackage{graphicx}
\usepackage{tcolorbox}

\graphicspath{{./graphics/}}

\title{Graph Similarity Based on Matrix Norms}

\author{Timo Gervens}{RWTH Aachen University,
  Germany}{gervens@informatik.rwth-aachen.de}{https://orcid.org/0000-0002-1224-9853}{This work was
  supported by the DFG grant GR 1492/20-1.}

\author{Martin Grohe}{RWTH Aachen University, Germany}{grohe@informatik.rwth-aachen.de}{https://orcid.org/0000-0002-0292-9142}{}

\authorrunning{T. Gervens and M. Grohe} 

\Copyright{Timo Gervens and Martin Grohe} 



\nolinenumbers

\EventEditors{John Q. Open and Joan R. Access}
\EventNoEds{2}
\EventLongTitle{42nd Conference on Very Important Topics (CVIT 2016)}
\EventShortTitle{CVIT 2016}
\EventAcronym{CVIT}
\EventYear{2016}
\EventDate{December 24--27, 2016}
\EventLocation{Little Whinging, United Kingdom}
\EventLogo{}
\SeriesVolume{42}
\ArticleNo{23}

\usepackage{cite}
\makeatletter
\newcommand{\citecomment}[2][]{\citen{#2}#1\citevar}
\newcommand{\citeone}[1]{\citecomment{#1}}
\newcommand{\citetwo}[2][]{\citecomment[,~#1]{#2}}
\newcommand{\citevar}{\@ifnextchar\bgroup{;~\citeone}{\@ifnextchar[{;~\citetwo}{]}}}
\newcommand{\citefirst}{\@ifnextchar\bgroup{\citeone}{\@ifnextchar[{\citetwo}{]}}}
\newcommand{\cites}{[\citefirst}
\makeatother

\usepackage{tikz}
\usetikzlibrary{positioning}

\newcommand{\refsec}[1]{Section~\ref{sec:#1}}

\renewcommand{\refeq}[1]{Equation~\ref{eq:#1}}
\newcommand{\reffig}[1]{Figure~\ref{fig:#1}}

\newcommand{\refthm}[1]{Theorem~\ref{thm:#1}}
\newcommand{\reflem}[1]{Lemma~\ref{lem:#1}}

\newcommand{\refdef}[1]{Definition~\ref{def:#1}}

\newcommand{\refapx}[1]{Appendix~\ref{apx:#1}}

\newcommand{\N}{\ensuremath{\mathbb{N}}}

\newcommand{\R}{\ensuremath{\mathbb{R}}}

\newcommand{\abs}[1]{\ensuremath{\left\vert{#1}\right\vert}}
\newcommand{\norm}[1]{\ensuremath{\left\lVert{#1}\right\rVert}}

\DeclareMathOperator*{\argmin}{argmin}
\DeclareMathOperator*{\argmax}{argmax}

\usepackage{xspace}

\newcommand*{\hamcycle}{\textsc{Ham-Cycle}\@\xspace}
\newcommand*{\threepartition}{\textsc{Three-Partition}\@\xspace}
\newcommand*{\maxcut}{\textsc{Max-Cut}\@\xspace}

\DeclareMathOperator*{\dist}{dist}
\newcommand*{\distp}{\dist\nolimits_p}

\newcommand*{\distn}[1]{\ensuremath{\dist\nolimits_{#1}}}
\DeclareMathOperator*{\cdist}{C-dist}
\newcommand*{\cdistp}{\cdist\nolimits_p}

\DeclareMathOperator*{\boundd}{bound}

\newcommand*{\fboundp}{\ensuremath{\boundd\nolimits_p}}

\newcommand*{\boundp}[1]{\ensuremath{\boundd\nolimits_p\hspace{-0.35ex}\left({#1}\right)}}
\newcommand*{\boundn}[2]{\ensuremath{\boundd\nolimits_{#1}\hspace{-0.3ex}\left({#2}\right)}}
\newcommand*{\mc}{\upshape{MC}\@\xspace}
\DeclareMathOperator*{\fmc}{MC}
\newcommand*{\mmc}{\upshape{MMC}\@\xspace}
\DeclareMathOperator*{\fmmc}{MMC}

\newenvironment{prblm}[1]{\begin{center}
    \ifthenelse{\equal{#1}{}}{\begin{tcolorbox}[width=\textwidth-4em,colback=black!10,boxrule=1pt]}{\begin{tcolorbox}[width=\textwidth-4em,title=#1,coltitle=white,colbacktitle=black!50,colback=black!10,boxrule=1pt]}
      \begin{description}
      }{
      \end{description}
    \end{tcolorbox}
  \end{center}
}

\newcommand{\nprob}[3]{
  \begin{prblm}{#1}
    \item[Instance:] #2
    \item[Problem:] #3
  \end{prblm}
}

\newcommand{\CG}{{\mathcal G}}
\newcommand{\CS}{{\mathcal S}}

\DeclareMathOperator{\sg}{sg}
\DeclareMathOperator{\Inj}{Inj}
\DeclareMathOperator{\abso}{abs}

\newcommand{\Nat}{{\mathbb N}}
\newcommand{\PNat}{{\mathbb N}_{>0}}
\newcommand{\Real}{{\mathbb R}}
\newcommand{\NNReal}{{\mathbb R}_{\ge0}}
\newcommand{\PReal}{{\mathbb R}_{>0}}

\newcommand*{\Distp}{$\textup{\textsc{Dist}}_p$\@\xspace}
\newcommand*{\Distf}{$\textup{\textsc{Dist}}_F$\@\xspace}
\newcommand*{\Dist}[1]{$\textup{\textsc{Dist}}_{#1}$\@\xspace}

\newcommand*{\CDistp}{$\textup{\textsc{Color-Dist}}_p$\@\xspace}

\newcommand{\NP}{\textsf{\upshape NP}}
\newcommand{\PTIME}{\textsf{\upshape P}}
\renewcommand{\vec}[1]{\boldsymbol{#1}}

\begin{document}

\maketitle

\begin{abstract}
  Quantifying the similarity between two graphs is a fundamental
  algorithmic problem at the heart of many data analysis tasks for
  graph-based data. In this paper, we study the computational
  complexity of a family of similarity measures based on quantifying
  the mismatch between the two graphs, that is, the ``symmetric
  difference'' of the graphs under an optimal alignment of the vertices. An important
  example is similarity based on graph edit distance. While edit
  distance calculates the ``global'' mismatch, that is, the number of
  edges in the symmetric difference, our main focus is on ``local''
  measures calculating the maximum mismatch per vertex.

Mathematically, our similarity measures are best expressed in terms of
the adjacency matrices: the mismatch between graphs is expressed as
the difference of their adjacency matrices (under an optimal
alignment), and we measure it by applying some matrix norm. Roughly
speaking, global measures like graph edit distance correspond to
entrywise matrix norms like the Frobenius norm and local measures
correspond to operator norms like the spectral norm.

We prove a number of strong NP-hardness and inapproximability results
even for very restricted graph classes such as bounded-degree trees.

\end{abstract}

\section{Introduction}
Graphs are basic models ubiquitous in science and engineering.  They
are used to describe a diverse range of objects and processes from
chemical compounds to social interactions. To understand and classify
graph models, we need to compare graphs. Since data and models are not
always guaranteed to be exact, it is essential to understand what
makes two graphs similar or dissimilar, and to be able to compute
similarity efficiently. There are many different approaches to
similarity, for example, based on edit-distance
(e.g.~\cite{ApproxGI,bun97,HardnessRGI}), spectral similarity (e.g.~\cite{tstmotkar+18,SpectralRGI,ume88}),
optimal transport (e.g.~\cite{BarbeSGBG20,MareticGCF19,Memoli11}), or
behavioral equivalence (e.g.\ \cite{gla01,san12}).
This is only natural, because the choice of a ``good'' similarity
measure will usually depend on the application. While graph
similarity has received considerable attention in application areas
such as computer vision (see, for example, \cite{ThirtyYears,fogperven14}) and network
science (see, for example, \cite{emmdehshi16}), theoretical computer
scientists have not explored similarity systematically; only
specific ``special cases'' such as isomorphism \cite{groschwe20}
and bisimilarity \cite{san12} have been studied to great depth. Yet it seems
worthwhile to develop a \emph{theory of graph similarity} that
compares different similarity measures, determines their algorithmic
and semantic properties, and thus gives us a better understanding of
their suitability for various kinds of applications. We see our paper
as one contribution to such a theory.

Maybe the simplest graph similarity measure is based on graph edit
distance: the \emph{edit distance} between graphs $G,H$ (for
simplicity of the same order) is the minimum number of edges that need
to be added and deleted from $G$ to obtain a graph ismorphic to $H$.

In this paper, we study the computational complexity of a natural
class of similarity measures that generalize the similarity derived
from edit distance of graphs. In general, we view similarity as
proximity with respect to some metric. A common way of converting a
graph metric $d$ into a similarity measure $s$ is to let
$s(G,H)\coloneqq\exp(-\beta\cdot d(G,H))$ for some constant
$\beta>0$. For our considerations the transformation between distance
and similarity is irrelevant, so we focus directly on the metrics. In
terms of computational complexity, computing similarity tends to be a
hard algorithmic problem. It is known that computing edit distance,
exactly or approximately, is \NP-hard
\cite{ApproxGI,Grohe,lin94,HardnessRGI} even on very restricted graph
classes. In fact, the problem is closely related to the quadratic
assignment problem \cite{makmansvi14,nagsvi09}, which is notorious for
being a very hard combinatorial optimization problem also for
practical, heuristic approaches. Within the spectrum of similarities,
the ``limit'' case of graph isomorphism shows that overall the
complexity of graph similarity is far from trivial.

The metrics we study in this paper
are based on minimizing the
{mismatch} between two graphs.
For graphs $G,H$ with the same
vertex set $V$, we define their \emph{mismatch graph} $G-H$ to be the
graph with vertex set $V$ and edge set $E(G)\triangle E(H)$, the symmetric
difference between the edge sets of the graphs. We assign \emph{signs}
to edges of the mismatch graph to indicate which graph they come from,
say, a positive sign to the edges in $E(G)\setminus E(H)$ and a
negative sign to the edges in $E(H)\setminus E(G)$. To quantify the
mismatch between our graphs we introduce a \emph{mismatch norm} $\mu$ on
signed graphs that satisfies a few basic axioms such as
subadditivity as well as invariance under permutations and under
flipping the signs of all edges. Now for graphs $G,H$, not
necessarily with the same vertex set, but for simplicity of the same
order,\footnote{A general definition that also applies to graphs of
  distinct order can be found in
Section~\ref{sec:mismatch}, but for the hardness results we prove in
this paper we can safely restrict our attention to graphs of the same
order; this only makes the results stronger.} we define the distance $\dist_\mu(G,H)$
to be the minimum of $\mu(G^\pi-H)$, where $\pi$ ranges over all
bijective mappings from $V(G)$ to $V(H)$ and $G^\pi$ is the image of
$G$ under $\pi$.

The simplest mismatch norm $\mu_{\textup{ed}}$ just counts the number
of edges of the mismatch graph $G-H$, ignoring their signs. Then the
associated distance $\dist_{\textup{ed}}(G,H)$ is the edit distance
between $G$ and $H$. (Note that we write $\dist_{\textup{ed}}$ instead
of the clunky $\dist_{\mu_{\textup{ed}}}$; we will do the same for
other mismatch norms discussed here.) Another simple yet interesting
mismatch norm is $\mu_{\textup{deg}}$ measuring the maximum degree of
the mismatch graph, again ignoring the signs of the edges. Then
$\dist_{\textup{deg}}(G,H)$ measures how well we can align the two
graphs in order to minimize the ``local mismatch'' at every
node. Hence an alignment where at every vertex there is a mismatch of
one edge yields a smaller $\dist_{\textup{deg}}$ than an alignment
that is perfect at all nodes expect one, where it has a mismatch of,
say, $n/2$, where $n$ is the number of vertices. For edit distance it
is the other way round. Depending on the application, one or the other
may be preferable.
Another well-known graph metric that can be described via
the mismatch graph is Lov\'asz's \emph{cut distance} (see
\cite[Chapter~8]{lov12} and Section~\ref{sec:mismatch} of this
paper). And, last but not least, for the mismatch norm
$\mu_{\textup{iso}}$ defined to be $0$ if the mismatch graph has no
edges and $1$ otherwise, $\dist_{\textup{iso}}(G,H)$ is $0$ if $G$ and
$H$ are isomorphic, and $1$ otherwise, so computing the distance
between two graphs amounts to deciding if they are isomorphic.

Mathematically, the framework of mismatch norms and the associated
distances is best described in terms of the adjacency matrices of the
graphs; the adjacency matrix $A_{G-H}$ of the mismatch graph (viewed
as a matrix with entries $0,+1,-1$ displaying the signs of the edges)
is just the difference $A_G-A_H$. Then mismatch norms essentially are
just matrix norms applied to to $A_{G-H}$. It turns out that ``global
norms'' such as edit distance and direct generalizations correspond to
entrywise matrix norms (obtained by applying a vector norm to the
flattened matrix), and ``local norms'' such as $\mu_{\textup{deg}}$
correspond to operator matrix norms (see
Section~\ref{sec:mismatch}). Cut distance corresponds to the cut norm
of matrices. Instead of adjacency matrices, we can also consider the
Laplacian matrices, exploiting that $L_{G-H}=L_G-L_H$, and obtain
another interesting family of graph distances.

For every mismatch norm $\mu$, we are interested in the problem
$\textsc{Dist}_\mu$ of computing $\dist_\mu(G,H)$ for two given graphs
$G,H$. Note that this is a minimization problem where the feasible
solutions are the bijective mappings between the vertex sets of the
two graphs. It turns out that the problem is hard for most mismatch
norms $\mu$, in particular almost all the natural choices discussed
above. The single exception is $\mu_{\text{iso}}$ related to the graph
isomorphism problem. Furthermore, the hardness results usually hold
even if the input graphs are restricted to be very simple, for example
trees or bounded degree graphs. For edit distance and the related
distance based on entrywise matrix norms this was already known
\cite{ApproxGI,Grohe}. Our focus in this paper is on operator
norms. We prove a number of hardness results for different graph
classes. One of the strongest ones is the following
(see~\refthm{MuliplicativeApproxThreshold}). Here \Distp\ denotes the
distance measure derived from the $\ell_p$-operator norm.

\begin{theorem}
  For $1\le p\le\infty$
  there is a constant $c>1$ such that, unless
  $\textsf{\upshape P}=\textsf{\upshape NP}$,
  \Distp\ has no factor-$c$ approximation algorithm even if the input
  graphs are restricted to be trees of bounded degree.
\end{theorem}

For details and additional results, we refer to Section~\ref{sec:ComplexityOperatorNorms} and \ref{sec:approx}.

Initially we aimed for a general hardness result that applies to all
mismatch norms satisfying some additional natural conditions. However,
we found that the hardness proofs, while following the same general
strategies, usually have some intricacies exploiting special
properties of the specific norms. Furthermore, for cut distance, none
of these strategies seemed to work. Nevertheless, we were able to give
a hardness proof for cut distance (Theorem~\ref{theo:cut}) that is
simply based on the hardness of computing the cut norm of the matrix
\cite{AlonNaor06}. This is remarkable in so far as usually the hard
part of computing the distance is to find an optimal alignment $\pi$,
whereas computing $\mu(G^\pi-H)$ is usually easy. For cut norm, it is
even hard to compute $\mu(G^\pi-H)$ for a fixed alignment $\pi$.

\subsection{Related Work}

Graph similarity has mostly been studied in specific application
areas, most importantly computer vision, data mining, and machine
learning (see the references above). Of course not all similarity
measures are based on mismatch. For example, metrics derived from
vector embedding or graph kernels in machine learning (see
\cite{krijohmor19}) provide a completely different approach (see
\cite{gro20c} for a broader discussion). Of interest compared to our
work (specifically for the $\ell_2$-operator norm a.k.a.\ spectral
norm) is the spectral approach proposed by Kolla, Koutis, Madan, and
Sinop \cite{SpectralRGI}. Intuitively, instead of the ``difference''
of two graphs that is described by our mismatch graphs, their
approach is based on taking a ``quotient''.

The complexity of similarity, or ``approximate graph isomorphism'', or
``robust graph isomorphism'' has been studied in
\cite{Arora2002,ApproxGI,Grohe, Keldenich,SpectralRGI ,lin94,HardnessRGI}, mostly based on graph edit
distance and small variations. Operator norms have not been considered
in this context.

\section{Preliminaries}
\label{sec:prel}
We denote the class of real numbers by $\Real$ and the nonnegative
 and positive
reals by $\NNReal,\PReal$, respectively. By $\Nat,\PNat$ we denote the
sets of nonnegative resp.\ positive integers. For every
$n\in\PNat$ we let $[n]\coloneqq\{1,\ldots,n\}$.

We will consider matrices with real entries and with rows and columns
indexed by arbitrary finite sets. Formally, for finite sets $V,W$, a
$V\times W$ matrix is a function $A\colon V\times W\to\Real$. A
standard $m\times n$-matrix is just an $[m]\times[n]$-matrix. We
denote the set of all $V\times W$ matrices by $\Real^{V\times W}$, and we
denote the entries of a matrix $A$ by $A_{vw}$.

For a square matrix $A\in\Real^{V\times
  V}$ and injective mapping $\pi:V\to W$, we
let $A^\pi$ be the $V^\pi\times V^\pi$-matrix with entries
$A_{v^\pi w^\pi}\coloneqq A_{vw}$. Note
that we apply the mapping $\pi$ from the right and denote the image
of $v$ under $\pi$ by $v^\pi$. If $\rho:W\to X$ is another mapping, we
denote the composition of $\pi$ and $\rho$ by $\pi\rho$. We typically
use this notation for mappings between matrices and between graphs.

We use a standard graph theoretic notation. We denote the vertex set
of a graph $G$ by $V(G)$ and the edge set by $E(G)$. Graphs are
finite, simple,
and undirected, that is, $E(G)\subseteq\binom{V(G)}{2}$. We denote
edges by $vw$ instead of $\{v,w\}$. The \emph{order} of a graph is
$|G|\coloneqq |V(G)|$. The \emph{adjacency matrix} $A_G\in\{0,1\}^{V(G)\times
V(G)}$ is defined in the usual way. 
We denote the class of all graphs by $\CG$.

Let $G$ be a graph with vertex set $V\coloneqq V(G)$. For a
mapping $\pi:V\to W$ we let $G^\pi$ be the
graph with vertex set $\{v^\pi\mid v\in V\}$ and edge set
$\{v^\pi w^\pi\mid vw\in E(G)\text{ with }v^\pi\neq w^\pi\}$. 
Then $A_{G^\pi}=A_G^\pi$.

For graphs $G,H$, we denote the set of all injective mappings
$\pi:V(G)\to V(H)$ by $\Inj(G,H)$. Graphs $G$ and $H$ are
\emph{isomorphic} (we write $G\cong H$) if there is an
$\pi\in\Inj(G,H)$ such that $G^\pi=H$. We think of the 
mappings in $\Inj(G,H)$, in particular if $|G|=|H|$ and they are
bijective, as \emph{alignments} between the graphs. Intuitively, to measure the
distance between two graphs, we will align them in an optimal way to
minimize the mismatch.

\section{Graph Metrics Based on Mismatches}
\label{sec:mismatch}
A \emph{graph metric} is a function $\delta:\CG\times\CG\to\Real_{\ge 0}$
such that
\begin{description}
\item[(GM0)] $\delta(G,H)=\delta(G',H')$ for all $G,G',H,H'$ such that $G\cong G'$ and 
  $H\cong H'$; 
\item[(GM1)] $\delta(G,G)=0$ for all $G$;
\item[(GM2)] $\delta(G,H)=\delta(H,G)$ for all $G,H$;
\item[(GM3)] $\delta(F,H)\le \delta(F,G)+\delta(G,H)$ for all $F,G,H$.
\end{description}
Note that we do not require $\delta(G,H)>0$ for all $G\neq H$, not
even for $G\not\cong H$, so strictly speaking this is
just a \emph{pseudometric}. We are interested in the complexity of the
following problem:

\nprob{\Dist{\delta}}{Graphs $G,H$, $p,q\in\PNat$}{Decide if $\delta(G,H)\ge\frac{p}{q}$}

A \emph{signed graph} is a weighted graph with edge weights $-1,+1$,
and for every edge $e$ of a signed graph we denote its \emph{sign} by
$\sg(e)$. For a signed graph ${\Delta}$, we let
$E_+({\Delta})\coloneqq\{e\in E({\Delta})\mid \sg(e)=+1\}$ and
$E_-({\Delta})\coloneqq\{e\in E({\Delta})\mid \sg(e)=-1\}$. \emph{Isomorphisms} of signed graphs
must preserve signs. A signed graph ${\Delta}$ is a
subgraph of a signed graph ${\Gamma}$ (we write ${\Delta}\subseteq {\Gamma}$) if
$V({\Delta})\subseteq V({\Gamma})$, $E_+({\Delta})\subseteq E_+({\Gamma})$, and
$E_-({\Delta})\subseteq E_-({\Gamma})$. We denote the class of all signed graphs by
$\CS$.

For every ${\Delta}\in\CS$, we let $-{\Delta}$ be the signed graph obtained from ${\Delta}$
by flipping the signs of all edges. We define the \emph{sum} of
${\Delta},{\Gamma}\in\CS$ to be the signed graph ${\Delta}+{\Gamma}$ with vertex set
$V({\Delta}+{\Gamma})=V({\Delta})\cup V({\Gamma})$ and signed edge sets
$E_+({\Delta}+{\Gamma})=\big(E_+({\Delta})\cup E_+({\Gamma})\big)\setminus\big(E_-({\Delta})\cup
E_-({\Gamma})\big)$ and
$E_-({\Delta}+{\Gamma})=\big(E_-({\Delta})\cup E_-({\Gamma})\big)\setminus\big(E_+({\Delta})\cup
E_+({\Gamma})\big)$.  The adjacency matrix $A_{\Delta}$ of a signed graph ${\Delta}$
displays the signs of the edges, so
$A_{\Delta}\in \{0,1,-1\}^{V({\Delta})\times V({\Delta})}$ with $(A_{\Delta})_{vw}=\sg(vw)$ if
$vw\in E({\Delta})$ and $(A_{\Delta})_{vw}=0$ otherwise.  Note that $A_{-{\Delta}}=-A_{\Delta}$
for all ${\Delta}\in\CS$ and $A_{{\Delta}+{\Gamma}}=A_{\Delta}+A_{\Gamma}$ for all ${\Delta},{\Gamma}\in\CS$ with
$V({\Delta})=V({\Gamma})$ and 
$E_+({\Delta})\cap E_+({\Gamma})=\emptyset$ and $E_-({\Delta})\cap E_-({\Gamma})=\emptyset$.

The \emph{mismatch graph} of two graphs $G,H$ is the signed graph
$G-H$ with vertex set $V(G-H)\coloneqq V(G)\cup V(H)$ and signed edge
set $E_+(G-H)\coloneqq E(G)\setminus E(H)$,
$E_-(G-H)\coloneqq E(H)\setminus E(G)$.  Note that if $V(G)=V(H)$ then
for the adjacency matrices we have $A_{G-H}=A_G-A_H$. Observe that
every signed graph ${\Delta}$ is the mismatch graph of the graphs
${\Delta}_+\coloneqq(V({\Delta}),E_+({\Delta}))$ and
${\Delta}_-\coloneqq(V({\Delta}),E_-({\Delta}))$.

A \emph{mismatch norm} is a function $\mu:\CS\to\NNReal$
satisfying the following conditions:
\begin{description}
\item[(MN0)] $\mu({\Delta})=\mu({\Gamma})$ for all ${\Delta},{\Gamma}\in\CS$ such that ${\Delta}\cong {\Gamma}$;
\item[(MN1)] $\mu({\Delta})=0$ for all ${\Delta}\in\CS$ with $E({\Delta})=\emptyset$;
\item[(MN2)] $\mu({\Delta})=\mu(-{\Delta})$ for all ${\Delta}$; 
\item[(MN3)] $\mu({\Delta}+{\Gamma})\le\mu({\Delta})+\mu({\Gamma})$ for all ${\Delta},{\Gamma}\in\CS$ with
  $V({\Delta})=V({\Gamma})$ and $E_+({\Delta})\cap E_+({\Gamma})=\emptyset$ and $E_-(S)\cap E_-({\Gamma})=\emptyset$.
\item[(MN4)] $\mu({\Delta})= \mu({\Gamma})$ for all ${\Delta},{\Gamma}\in\CS$ with 
  $E_+({\Delta})=E_+({\Gamma})$ and
  $E_-({\Delta})=E_-({\Gamma})$;
\end{description}
For every mismatch norm $\mu$, we define $\dist_\mu:\CG\times\CG\to\NNReal$
by
\[
  {\dist}_\mu(G,H)\coloneqq
  \begin{cases}
    \min_{\pi\in\Inj(G,H)}\mu(G^\pi-H)&\text{if }|G|\le|H|,\\
    \min_{\pi\in\Inj(H,G)}\mu(G-H^\pi)&\text{if }|H|<|G|.
  \end{cases}
\]
We write \Dist{\mu} instead of \Dist{\dist_\mu} to denote the
algorithmic problem of computing $\dist_\mu$ for two graphs $G,H$.

\begin{lemma}\label{lem:metric}
  For every mismatch norm $\mu$ the function $\dist_\mu$ is a graph
  metric. 
\end{lemma}

\begin{proof}
  Conditions (GM0), (GM1), (GM2) follow from (MN0), (MN1), (MN2),
  respectively. To prove (GM3), let $F,G,H$ be graphs. Padding the
  graphs with isolated vertices, by (MM4) we may assume that
  $|F|=|G|=|H|$. By (MN0) we may further assume that
  $V(F)=V(G)=V(H)\coloneqq V$.
  Choose
  $\pi,\rho\in\Inj(V,V)$ such that
  $\dist_\mu(F,G)=\mu(F^{\pi}-G)$ and
  $\dist_\mu(G,H)=\mu(G^{\rho}-H)$.

  Then by (MN0) we have
    \[
      \mu(F^{\pi\rho}-G^{\rho})=\mu\big((F^\pi-G)^{\rho}\big)=\mu(F^\pi-G)={\dist}_\mu(F,G).
    \]
    Now consider the two mismatch graphs ${\Delta}\coloneqq
    F^{\pi\rho}-G^\rho$ and ${\Gamma}\coloneqq G^\rho-H$. We have $E_+({\Delta})=E(F^{\pi\rho})\setminus
    E(G^\rho)$ and $E_+({\Gamma})=E(G^\rho)\setminus E(H)$. Thus $E_+({\Delta})\cap
    E(G^\rho)=\emptyset$ and $E_+({\Gamma})\subseteq
    E(G^\rho)$, which implies $E_+({\Delta})\cap
    E_+({\Gamma})=\emptyset$. Similarly, $E_-({\Delta})\cap
    E_-({\Gamma})=\emptyset$. Moreover, ${\Delta}+{\Gamma}=F^{\pi\rho}-H$, because
    \[
      A_{{\Delta}+{\Gamma}}=A_{\Delta}+A_{\Gamma}=(A_{F}^{\pi\rho}-A_{G}^\rho)+(A_{G}^\rho-A_{H})=A_{F}^{\pi\rho}-A_{H}=A_{F^{\pi\rho}-H}.
    \]
    Thus by (MN3),
    \[
      {\dist}_\mu(F,H)\le\mu(F^{\pi\rho}-H)=\mu({\Delta}+{\Gamma})\le\mu({\Delta})+\mu({\Gamma})={\dist}_\mu(F,G)+{\dist}_\mu(G,H).
    \]
    This proves (GM3).
  \end{proof}

\begin{remark}
  None of the five conditions (MN0)--(MN4) on a mismatch norm $\mu$
  can be dropped if we want to guarantee that $\dist_\mu$ is a graph
  metric, but of course we could replace them by other
  conditions. While (MN0)--(MN3) are very natural and directly
  correspond to conditions (GM0)--(GM3) for graph metrics, condition
  (MN4) is may be less so. We chose (MN4) as the simplest condition
  that allows us to compare graphs of different sizes.

  Having said this, it is worth noting that (MN0), (MN1) and (MN3)
  imply (MN4) for graphs ${\Delta},{\Gamma}$ with $|{\Delta}|=|{\Gamma}|$. For graphs ${\Delta},{\Gamma}$ with
  $|{\Delta}|<|{\Gamma}|$ they only imply $\mu({\Delta})\ge \mu({\Gamma})$. Thus as long as we
  only compare graphs of the same order, (MN4) is not needed. In
  particular, since our hardness results always apply to graphs of the
  same order, (MN4) is inessential for the rest of the paper.

  However, it is possible to replace (MN4) by other natural
  conditions. For example, Lov\'asz's metric based on a scaled cut
  norm \cite{Lovasz12} does not satisfy (MN4) and instead uses a
  blowup of graphs to a common size to compare graphs of different
  sizes.
\end{remark}

Let us now consider a few examples of mismatch norms.

\begin{example}[Isomorphism]\label{exa:iso-norm}
  The mapping $\iota\coloneqq\CS\to\NNReal$ defined by
  $\iota({\Delta})\coloneqq0$ if $E({\Delta})=\emptyset$ and $\iota({\Delta})\coloneqq1$
  otherwise is a mismatch norm. Under the metric $\dist_\iota$, all
  nonisomorphic graphs have distance $1$ (and isomorphic graphs have
  distance $0$, as they have under all graph metrics).
\end{example}

\begin{example}[Matrix Norms]\label{exa:matrix-norms}
  Recall that a matrix (pseudo) norm $\|\cdot\|$ associates with every
  matrix $A$ (say, with real entries) an
  nonnegative real $\|A\|$ in such a way that $\|N\|=0$ for matrices
  $N$ with only $0$-entries, $\|aA\|=|a|\cdot\|A\|$ for all matrices $A$
  and reals $a\in\Real$, and $\|A+B\|\le\|A\|+\|B\|$ for all
  matrices $A,B$ of the same dimensions.

  Actually, we are only interested in square matrices here. We call
  a matrix norm $\|\cdot\|$ \emph{permutation invariant} if for all $A\in\Real^{V\times V}$
  and all bijective $\pi:V\to V$ we have $\|A\|=\|A^\pi\|$.
  It is easy to see that for every permutation invariant matrix norm $\|\cdot\|$, the
  mapping $\mu_{\|\cdot\|}:\CS\to\NNReal$ defined by $\mu_{\|\cdot\|}({\Delta})\coloneqq\|A_{\Delta}\|$
  satisfies (MN0)--(MN3).

  We call a permutation invariant matrix norm $\|\cdot\|$ \emph{paddable} if it is
  invariant under extending matrices by zero entries, that is,
  $\|A\|=\|A'\|$ for all
  $A\in\Real^{V\times V}$, $A'\in\Real^{V'\times V'}$ such that
  $V'\supseteq V$, $A_{vw}=A'_{vw}$ for all $v,w\in V$, and
  $A'_{v'w'}=0$ if $v\in V'\setminus V$ or $w\in V'\setminus V$. A
  paddable matrix norm also satisfies (MN4).

  The following common matrix norms are paddable (and thus by
  definition also invariant). Let $1\leq p\leq\infty$ and $A\in\Real^{V\times
    V}$. 
  \begin{enumerate}
  \item \emph{Entrywise $p$-norm:}
    $\|A\|_{(p)}\coloneqq\left(\sum_{v,w\in
        V}|A_{vw}|^p\right)^{\frac{1}{p}}$. The best-known special case is the
      \emph{Frobenius norm} $\|\cdot\|_F\coloneqq\|\cdot\|_{(2)}$.
  \item \emph{$\ell_p$-operator norm:} $\|A\|_{p}\coloneqq\sup_{\vec
      x\in\Real^V\setminus\{\vec 0\}}\frac{\|A\vec x\|_p}{\|\vec x\|_p}$, where the
    vector $p$-norm is defined by $\|\vec
    a\|_p\coloneqq\left(\sum_{v\in V}a_v^p\right)^{\frac{1}{p}}$. In
    particular, $\|A\|_2$ is known as the \emph{spectral norm}.
 \item \emph{Absolute $\ell_p$-operator norm:}
   $\|A\|_{|p|}\coloneqq\|\abso(A)\|_p$, where $\abso$ takes entrywise
   absolute values. For the mismatch norm, taking entrywise absolute
   values means that we ignore the signs of the edges.
  \item \emph{Cut norm:} $\|A\|_{\square}\coloneqq\max_{S,T\subseteq
      V}\left|\sum_{v\in S,w\in T}A_{vw}\right|$.
  \end{enumerate}
\end{example}

\begin{example}[Laplacian Matrices]
  Recall that the \emph{Laplacian matrix} of a weighted graph $G$ with
  vertex set $V\coloneqq V(G)$ is the $V\times V$ matrix $L_G$ with
  off-diagonal entries $(L_G)_{vw}$ being the negative weight of the
  edge $vw\in E(G)$ if there is such an edge and $0$ otherwise and
  diagonal entries $(L_G)_{vv}$ being the sum of the weights of all
  edges incident with $v$. For an unweighted graph we have
  $L_G=D_G-A_G$, where $D_G$ is the diagonal matrix with the vertex
  degrees as diagonal entries.

  Observe that for a signed graph ${\Delta}=G-H$ we have $L_{\Delta}=L_G-L_H$.

  It is easy to see that for every paddable matrix norm $\|\cdot\|$,
  the function $\mu^L_{\|\cdot\|}:\CS\to\NNReal$ defined by
  $\mu^L_{\|\cdot\|}({\Delta})\coloneqq\|L_{\Delta}\|$ is a mismatch norm.
\end{example}

Clearly, \Dist{\mu} is not a hard problem for every mismatch
norm. For example, \Dist{\nu} is trivial for the trivial mismatch
norm $\nu$ defined by $\nu({\Delta})\coloneqq 0$ for all ${\Delta}$, and
\Dist{\iota} for $\iota$ from Example~\ref{exa:iso-norm} is
equivalent to the graph isomorphism problem and hence in
quasipolynomial time \cite{Babai16}.

However, for most natural matrix norms the associated graph distance
problem is \NP-hard. In particular, for every $p\in\PReal$,
this holds for the metric $\distn{(p)}$
based on the entrywise $p$-norm $\|\cdot\|_{(p)}$.

\begin{theorem}[\cite{Grohe}]
	For $p\in\PReal$, \Dist{(p)} is \NP-hard even if restricted to trees or bounded-degree graphs.
\end{theorem}

The proof in \cite{Grohe} is only given for the Frobenius norm, that
is, \Dist{(2)}, but it
actually applies to all $p$.
In the rest of the paper, we study the
complexity of \Distp, \Dist{|p|} and \Dist{\square}.

\section{Complexity for Operator Norms}
\label{sec:ComplexityOperatorNorms}
In this section we investigate the complexity of \Distp and \Dist{\abs{p}} for $1\leq p\leq\infty$.
However, as $\mu_{|p|}$ would only be a special case within the upcoming proofs, we omit to mention it explicitly.
We also omit to specify the possible values for $p$.

Given graphs $G$ and $H$, an alignment from $G$ to $H$, and a node $v\in V(G)$ we refer to the nodes whose adjacency is not preserved by $\pi$ as the \emph{$\pi$-mismatches of $v$}.
If $\pi$ is clear from the context we might omit it.
We call $G^\pi-H$ the \emph{mismatch graph of $\pi$}.

For all $\ell_p$-operator norms the value of $\mu_p(G^\pi-H)$ strongly depends on the maximum degree in the mismatch graph of $\pi$.
We capture this property with the following definition.

\begin{definition}
	Let $G$, $H$ be graphs of the same order and $\pi$ an alignment from $G$ to $H$.
	The \emph{$\pi$-mismatch count ($\pi$-\mc)} of a node $v\in V(G)$ is defined as:
	\[
	\fmc(v,\pi)\coloneqq\abs{\{w\in V(G)\,|\,w\text{ is a $\pi$-mismatch of }v\}}.
	\]
	We use \mc for nodes in $H$ analogously.
	The \emph{maximum mismatch count ($\pi$-\mmc)} of $\pi$ is defined as:
	\[
	\fmmc(\pi)\coloneqq\max_{v\in V(G)}\fmc(v,\pi).
	\]
\end{definition}

Again we might drop the $\pi$ if it is clear from the context.
Note that the \mmc corresponds to the maximum degree in the mismatch graph and that we use a slightly abbreviated notation in which we assume the graphs are given by the alignment.

The $\ell_1$-operator norm and the $\ell_\infty$-operator norm measure exactly the maximum mismatch count.
Due to the relation between the $\ell_p$-operator norms we can derive an upper bound for $\mu_p$.
The proof of \reflem{MMCupper2} can be found in \refapx{ComplexityOperatorNorms}.

\begin{lemma}
	\label{lem:MMCupper2}
	Let $G$, $H$ be graphs of the same order and $\pi$ an alignment from $G$ to $H$.
	Then
	\begin{align*}
	\mu_1(G^\pi-H)=\mu_\infty(G^\pi-H)=\fmmc(\pi),\\	
	\mu_p(G^\pi-H)\leq\fmmc(\pi).
	\end{align*}
\end{lemma}

Next, we observe that $\mu_p$ is fully determined by the connected component of the mismatch graph with the highest mismatch norm.
The proof of \reflem{MMR} can be found in \refapx{ComplexityOperatorNorms}.

\begin{lemma}
	\label{lem:MMR}
	Let $G$, $H$ be graphs of the same order, $\pi$ an alignment from $G$ to $H$, and $\mathcal{C}$ the set of all connected components in $G^\pi-H$.
	Then
	\begin{align*}
		\mu_p(G^\pi-H)&=\max_{C\in\mathcal{C}}\mu_p(C).
	\end{align*}	
\end{lemma}

For the sake of readability, we introduce a function as abbreviation for our upcoming bounds.

\begin{definition}
	For all $1\leq p\leq\infty$ we define the function $\fboundp$ as follows:
	\begin{align*}
		\boundp{c}&\coloneqq\max\left(c^{1/p},c^{1-1/p}\right).
	\end{align*}
\end{definition}
In particular, $\boundn{2}{c}=\sqrt{c}$.
Now we derive our lower bound.
The proof of \reflem{MCG} can be found in \refapx{ComplexityOperatorNorms}.
\begin{lemma}
	\label{lem:MCG}
	Let $G$, $H$ be graphs of the same order, $\pi$ an alignment from $G$ to $H$, and $v\in V(G)$.\\
	If $v$ has at least $c$ mismatches, then
	\[\mu_p(G^\pi-H)\geq\boundp{c}.\]
	If $G^\pi-H$ is a star, then
	\[\mu_p(G^\pi-H)=\boundp{\fmmc(\pi)}.\]
\end{lemma}

While $\mu_1$ and $\mu_\infty$ simply measure the \mmc, $\mu_2$ also considers the connectedness of the mismatches around the node with the highest \mc.
As \reflem{MMCupper2} and \reflem{MCG} tell us, $\mu_2$ ranges between the $\sqrt{\fmmc}$ and the \mmc.
In fact, the lower bound is tight for stars and the upper bound is tight for regular graphs.
This is intuitive as these are the extreme cases in which no mismatch can be removed/added without decreasing/increasing the \mmc, respectively. 
Other $\ell_p$-operator norms interpolate between $\mu_2$ and $\mu_1$ / $\mu_\infty$ in terms of how much they value the connectedness within the mismatch component of the node with the highest \mc.

The lower bound for $\mu_p(G^\pi-H)$ gives us lower bound for $\dist_p(G,H)$.
\begin{lemma}
	\label{lem:MMClowerG}
	Let $G$, $H$ be graphs of the same order and $\pi$ an alignment from $G$ to $H$.
	If all alignments from $G$ to $H$ have a node with at least $c$ mismatches, then
	\begin{align*}
		\distp(G,H)&\geq\boundp{c}.
	\end{align*}
	\begin{proof}
		This follows directly from the first claim of \reflem{MCG}.
	\end{proof}
\end{lemma}

The following upper bound might seem to have very restrictive conditions but is actually used in several hardness proofs.
\begin{lemma}
	\label{lem:MMCupperG}
	Let $G$, $H$ be graphs of the same order and $\pi$ an alignment from $G$ to $H$.
	If the mismatch graph of $\pi$ consists only of stars, then
	\begin{alignat*}{3}
		\distp(G,H)&\;\leq\;&\mu_p(G^\pi-H)&\;=\;&\boundp{\fmmc(\pi)}.
	\end{alignat*}
	\begin{proof}
		This follows directly from \reflem{MMR} and the second claim of \reflem{MCG}.
	\end{proof}
\end{lemma}

The last tool we need to prove the hardness is that we can distinguish two alignments by their \mmc as long as the mismatch graph of the alignment with the lower \mmc consists only of stars.
The proof of \reflem{MMCDistinguish} can be found in \refapx{ComplexityOperatorNorms}.
\begin{lemma}
	\label{lem:MMCDistinguish}
	For all $c,d\in\N$ with $c<d$ it holds that $\boundp{c}<\boundp{d}$.
\end{lemma}

The graph isomorphism problem becomes solvable in polynomial time if restricted to graphs of bounded degree \cite{Luks}.
In contrast to this, \Distf is \NP-hard even under this restriction \cite{Grohe}.
We show that the same applies to \Distp.

The reduction in the hardness proof works for any mismatch norm which can, said intuitively, distinguish the mismatch norm of the 1-regular graph of order $n$ from any other $n$-nodes mismatch graph in which every node has at least one $-1$ edge but at least one node has an additional $-1$ edge and $+1$ edge.
In particular, the construction also works for \Dist{(p)}.
However, it does not work for the cut-distance, for which we independently prove the hardness in \refsec{CutNorm}.
\begin{theorem}
	\label{thm:NPhard}
	\Distp and \Dist{|p|} are \NP-hard for $1\leq p\leq\infty$ even if both graphs have bounded degree.
	\begin{proof}
		The proof is done by reduction from the \NP-hard Hamiltonian Cycle problem in 3-regular graphs (\hamcycle) \cite{Intractability}.
		Given a 3-regular graph $G$ of order $n$ as an instance of \hamcycle, the reduction uses the $n$-nodes cycle $C_n$ and $G$ as inputs for \Distp.
		We claim $G$ has a Hamiltonian cycle if and only if $\distp(C_n,G)\leq\boundp{1}$.
		
		Assume that $G$ has a Hamiltonian cycle.
		Then there exists a bijection $\pi:V(C_n)\rightarrow V(G)$ that aligns the cycle $C_n$ perfectly with the Hamiltonian cycle in $G$.
		Each node in $G$ has three neighbors, two of which are matched correctly by $\pi$ as they are part of the Hamiltonian cycle.
		Hence, each node has a $\pi$-mismatch count of 1 and obviously the \mmc of $\pi$ is 1 as well.
		According to \reflem{MMCupperG}, we get $\distp(C_n,G)\leq\boundp{1}$.
		
		Conversely, assume that $G$ has no Hamiltonian cycle.
		Then, for any alignment $\pi'$ from $C_n$ to $G$, there exists at least one edge $vw$ in $C_n$ that is not mapped to an edge in $G$.
		Hence, only one of the three edges incident to $\pi'(v)$ in $G$ can be matched correctly.
		In total, $v$ has at least one mismatch from $C_n$ to $G$ and two mismatches from $G$ to $C_n$, which implies $\fmc(v,\pi')\geq3$.
		Using \reflem{MMClowerG}, we get $\distp(C_n,G)\geq\boundp{3}$.
		And then $\distp(C_n,G)>\boundp{1}$ according to \reflem{MMCDistinguish}.
	\end{proof}
\end{theorem}

Next, we modify the construction to get an even stronger \NP-hardness result.
The proof can be found in \refapx{ComplexityOperatorNorms}.
\begin{theorem}
	\label{thm:NPhardPath}
	\Distp and \Dist{|p|} are \NP-hard for $1\leq p\leq\infty$ even if restricted to a path and a graph of maximum degree $3$.
\end{theorem}

Similar to bounded degree input graphs, restricting graph isomorphism to trees allows it to be solved in polynomial time \cite{SubtreeIsomorphism} but \Distf is \NP-hard for trees \cite{Grohe}.
We show that \Distp remains \NP-hard for trees even when applying the bounded degree restriction simultaneously.
\begin{theorem}
	\label{thm:NPhardBoundedTree}
	\Distp and \Dist{|p|} are \NP-hard for $1\leq p\leq\infty$ even if restricted to bounded-degree trees.
\end{theorem}

\begin{figure*}[htb!]
	\centering
	\begin{subfigure}[t]{0.25\textwidth}
		\centering
		\resizebox{\textwidth}{!}{\begin{tikzpicture}[
every node/.style={draw,shape=circle,fill=red}]

\node[fill=black](px11){};
\node[below=1.4cm of px11, fill=black](px12){};
\node[below=1.4cm of px12, fill=black](px13){};
\node[below=1.4cm of px13, fill=black](px14){};

\node[above left=0.27cm and 0.4cm of px11, fill=red](ex111){};
\node[left=0.29cm of px11, fill=red](ex112){};
\node[below left=0.27cm and 0.4cm of px11, fill=red](ex113){};
\node[right=0.29cm of px11, fill=orange](ex114){};

\node[above left=0.27cm and 0.4cm of px12, fill=red](ex121){};
\node[left=0.29cm of px12, fill=red](ex122){};
\node[below left=0.27cm and 0.4cm of px12, fill=red](ex123){};
\node[right=0.29cm of px12, fill=orange](ex124){};
\node[above right=0.27cm and 0.4cm of px12, fill=pink](ex125){};
\node[below right=0.27cm and 0.4cm of px12, fill=pink](ex126){};
\node[right=0.29cm of ex125, fill=blue](ex127){};
\node[right=0.29cm of ex126, fill=blue](ex128){};

\node[above left=0.27cm and 0.4cm of px13, fill=red](ex131){};
\node[left=0.29cm of px13, fill=red](ex132){};
\node[below left=0.27cm and 0.4cm of px13, fill=red](ex133){};
\node[right=0.29cm of px13, fill=orange](ex134){};
\node[above right=0.27cm and 0.4cm of px13, fill=pink](ex135){};
\node[below right=0.27cm and 0.4cm of px13, fill=pink](ex136){};
\node[right=0.29cm of ex135, fill=blue](ex137){};
\node[right=0.29cm of ex136, fill=blue](ex138){};

\node[above left=0.27cm and 0.4cm of px14, fill=red](ex141){};
\node[left=0.29cm of px14, fill=red](ex142){};
\node[below left=0.27cm and 0.4cm of px14, fill=red](ex143){};
\node[right=0.29cm of px14, fill=orange](ex144){};

\node[below=1.4cm of px14, fill=black](px21){};
\node[below=1.4cm of px21, fill=black](px22){};
\node[below=1.4cm of px22, fill=black](px23){};

\node[above left=0.27cm and 0.4cm of px21, fill=red](ex211){};
\node[left=0.29cm of px21, fill=red](ex212){};
\node[below left=0.27cm and 0.4cm of px21, fill=red](ex213){};
\node[right=0.29cm of px21, fill=orange](ex214){};

\node[above left=0.27cm and 0.4cm of px22, fill=red](ex221){};
\node[left=0.29cm of px22, fill=red](ex222){};
\node[below left=0.27cm and 0.4cm of px22, fill=red](ex223){};
\node[right=0.29cm of px22, fill=orange](ex224){};
\node[above right=0.27cm and 0.4cm of px22, fill=pink](ex225){};
\node[below right=0.27cm and 0.4cm of px22, fill=pink](ex226){};
\node[right=0.29cm of ex225, fill=blue](ex227){};
\node[right=0.29cm of ex226, fill=blue](ex228){};

\node[above left=0.27cm and 0.4cm of px23, fill=red](ex231){};
\node[left=0.29cm of px23, fill=red](ex232){};
\node[below left=0.27cm and 0.4cm of px23, fill=red](ex233){};
\node[right=0.29cm of px23, fill=orange](ex234){};

\node[below=1.4cm of px23, fill=black](px31){};
\node[below=1.4cm of px31, fill=black](px32){};
\node[below=1.4cm of px32, fill=black](px33){};

\node[above left=0.27cm and 0.4cm of px31, fill=red](ex311){};
\node[left=0.29cm of px31, fill=red](ex312){};
\node[below left=0.27cm and 0.4cm of px31, fill=red](ex313){};
\node[right=0.29cm of px31, fill=orange](ex314){};

\node[above left=0.27cm and 0.4cm of px32, fill=red](ex321){};
\node[left=0.29cm of px32, fill=red](ex322){};
\node[below left=0.27cm and 0.4cm of px32, fill=red](ex323){};
\node[right=0.29cm of px32, fill=orange](ex324){};
\node[above right=0.27cm and 0.4cm of px32, fill=pink](ex325){};
\node[below right=0.27cm and 0.4cm of px32, fill=pink](ex326){};
\node[right=0.29cm of ex325, fill=blue](ex327){};
\node[right=0.29cm of ex326, fill=blue](ex328){};

\node[above left=0.27cm and 0.4cm of px33, fill=red](ex331){};
\node[left=0.29cm of px33, fill=red](ex332){};
\node[below left=0.27cm and 0.4cm of px33, fill=red](ex333){};
\node[right=0.29cm of px33, fill=orange](ex334){};

\node[right=3.2cm of px11, fill=black](px41){};
\node[below=1.4cm of px41, fill=black](px42){};
\node[below=1.4cm of px42, fill=black](px43){};
\node[below=1.4cm of px43, fill=black](px44){};

\node[above right=0.27cm and 0.4cm of px41, fill=red](ex411){};
\node[right=0.29cm of px41, fill=red](ex412){};
\node[below right=0.27cm and 0.4cm of px41, fill=red](ex413){};
\node[left=0.29cm of px41, fill=orange](ex414){};

\node[above right=0.27cm and 0.4cm of px42, fill=red](ex421){};
\node[right=0.29cm of px42, fill=red](ex422){};
\node[below right=0.27cm and 0.4cm of px42, fill=red](ex423){};
\node[left=0.29cm of px42, fill=orange](ex424){};
\node[above left=0.27cm and 0.4cm of px42, fill=pink](ex425){};
\node[below left=0.27cm and 0.4cm of px42, fill=pink](ex426){};
\node[left=0.29cm of ex425, fill=blue](ex427){};
\node[left=0.29cm of ex426, fill=blue](ex428){};

\node[above right=0.27cm and 0.4cm of px43, fill=red](ex431){};
\node[right=0.29cm of px43, fill=red](ex432){};
\node[below right=0.27cm and 0.4cm of px43, fill=red](ex433){};
\node[left=0.29cm of px43, fill=orange](ex434){};
\node[above left=0.27cm and 0.4cm of px43, fill=pink](ex435){};
\node[below left=0.27cm and 0.4cm of px43, fill=pink](ex436){};
\node[left=0.29cm of ex435, fill=blue](ex437){};
\node[left=0.29cm of ex436, fill=blue](ex438){};

\node[above right=0.27cm and 0.4cm of px44, fill=red](ex441){};
\node[right=0.29cm of px44, fill=red](ex442){};
\node[below right=0.27cm and 0.4cm of px44, fill=red](ex443){};
\node[left=0.29cm of px44, fill=orange](ex444){};

\node[below=1.4cm of px44, fill=black](px51){};
\node[below=1.4cm of px51, fill=black](px52){};
\node[below=1.4cm of px52, fill=black](px53){};

\node[above right=0.27cm and 0.4cm of px51, fill=red](ex511){};
\node[right=0.29cm of px51, fill=red](ex512){};
\node[below right=0.27cm and 0.4cm of px51, fill=red](ex513){};
\node[left=0.29cm of px51, fill=orange](ex514){};

\node[above right=0.27cm and 0.4cm of px52, fill=red](ex521){};
\node[right=0.29cm of px52, fill=red](ex522){};
\node[below right=0.27cm and 0.4cm of px52, fill=red](ex523){};
\node[left=0.29cm of px52, fill=orange](ex524){};
\node[above left=0.27cm and 0.4cm of px52, fill=pink](ex525){};
\node[below left=0.27cm and 0.4cm of px52, fill=pink](ex526){};
\node[left=0.29cm of ex525, fill=blue](ex527){};
\node[left=0.29cm of ex526, fill=blue](ex528){};

\node[above right=0.27cm and 0.4cm of px53, fill=red](ex531){};
\node[right=0.29cm of px53, fill=red](ex532){};
\node[below right=0.27cm and 0.4cm of px53, fill=red](ex533){};
\node[left=0.29cm of px53, fill=orange](ex534){};

\node[below=1.4cm of px53, fill=black](px61){};
\node[below=1.4cm of px61, fill=black](px62){};
\node[below=1.4cm of px62, fill=black](px63){};

\node[above right=0.27cm and 0.4cm of px61, fill=red](ex611){};
\node[right=0.29cm of px61, fill=red](ex612){};
\node[below right=0.27cm and 0.4cm of px61, fill=red](ex613){};
\node[left=0.29cm of px61, fill=orange](ex614){};

\node[above right=0.27cm and 0.4cm of px62, fill=red](ex621){};
\node[right=0.29cm of px62, fill=red](ex622){};
\node[below right=0.27cm and 0.4cm of px62, fill=red](ex623){};
\node[left=0.29cm of px62, fill=orange](ex624){};
\node[above left=0.27cm and 0.4cm of px62, fill=pink](ex625){};
\node[below left=0.27cm and 0.4cm of px62, fill=pink](ex626){};
\node[left=0.29cm of ex625, fill=blue](ex627){};
\node[left=0.29cm of ex626, fill=blue](ex628){};

\node[above right=0.27cm and 0.4cm of px63, fill=red](ex631){};
\node[right=0.29cm of px63, fill=red](ex632){};
\node[below right=0.27cm and 0.4cm of px63, fill=red](ex633){};
\node[left=0.29cm of px63, fill=orange](ex634){};

\draw
(px11) -- (px12)
(px12) -- (px13)
(px13) -- (px14)

(px11) -- (ex111)
(px11) -- (ex112)
(px11) -- (ex113)
(px11) -- (ex114)

(px12) -- (ex121)
(px12) -- (ex122)
(px12) -- (ex123)
(px12) -- (ex124)
(px12) -- (ex125)
(px12) -- (ex126)
(ex125) -- (ex127)
(ex126) -- (ex128)

(px13) -- (ex131)
(px13) -- (ex132)
(px13) -- (ex133)
(px13) -- (ex134)
(px13) -- (ex135)
(px13) -- (ex136)
(ex135) -- (ex137)
(ex136) -- (ex138)

(px14) -- (ex141)
(px14) -- (ex142)
(px14) -- (ex143)
(px14) -- (ex144)

(ex144) -- (ex214)
(px21) -- (px22)
(px22) -- (px23)

(px21) -- (ex211)
(px21) -- (ex212)
(px21) -- (ex213)
(px21) -- (ex214)

(px22) -- (ex221)
(px22) -- (ex222)
(px22) -- (ex223)
(px22) -- (ex224)
(px22) -- (ex225)
(px22) -- (ex226)
(ex225) -- (ex227)
(ex226) -- (ex228)

(px23) -- (ex231)
(px23) -- (ex232)
(px23) -- (ex233)
(px23) -- (ex234)

(ex234) -- (ex314)
(px31) -- (px32)
(px32) -- (px33)

(px31) -- (ex311)
(px31) -- (ex312)
(px31) -- (ex313)
(px31) -- (ex314)

(px32) -- (ex321)
(px32) -- (ex322)
(px32) -- (ex323)
(px32) -- (ex324)
(px32) -- (ex325)
(px32) -- (ex326)
(ex325) -- (ex327)
(ex326) -- (ex328)

(px33) -- (ex331)
(px33) -- (ex332)
(px33) -- (ex333)
(px33) -- (ex334)

(px41) -- (px42)
(px42) -- (px43)
(px43) -- (px44)

(px41) -- (ex411)
(px41) -- (ex412)
(px41) -- (ex413)
(px41) -- (ex414)

(px42) -- (ex421)
(px42) -- (ex422)
(px42) -- (ex423)
(px42) -- (ex424)
(px42) -- (ex425)
(px42) -- (ex426)
(ex425) -- (ex427)
(ex426) -- (ex428)

(px43) -- (ex431)
(px43) -- (ex432)
(px43) -- (ex433)
(px43) -- (ex434)
(px43) -- (ex435)
(px43) -- (ex436)
(ex435) -- (ex437)
(ex436) -- (ex438)

(px44) -- (ex441)
(px44) -- (ex442)
(px44) -- (ex443)
(px44) -- (ex444)

(ex444) -- (ex514)
(px51) -- (px52)
(px52) -- (px53)

(px51) -- (ex511)
(px51) -- (ex512)
(px51) -- (ex513)
(px51) -- (ex514)

(px52) -- (ex521)
(px52) -- (ex522)
(px52) -- (ex523)
(px52) -- (ex524)
(px52) -- (ex525)
(px52) -- (ex526)
(ex525) -- (ex527)
(ex526) -- (ex528)

(px53) -- (ex531)
(px53) -- (ex532)
(px53) -- (ex533)
(px53) -- (ex534)

(ex534) -- (ex614)
(px61) -- (px62)
(px62) -- (px63)

(px61) -- (ex611)
(px61) -- (ex612)
(px61) -- (ex613)
(px61) -- (ex614)

(px62) -- (ex621)
(px62) -- (ex622)
(px62) -- (ex623)
(px62) -- (ex624)
(px62) -- (ex625)
(px62) -- (ex626)
(ex625) -- (ex627)
(ex626) -- (ex628)

(px63) -- (ex631)
(px63) -- (ex632)
(px63) -- (ex633)
(px63) -- (ex634)

(ex334) -- (ex634)
;

\end{tikzpicture}
 		}		
		\caption{$T_1$}
	\end{subfigure}~
	\hspace{0.5cm}
	~ 
	\begin{subfigure}[t]{0.25\textwidth}
		\centering
		\resizebox{\textwidth}{!}{\begin{tikzpicture}[
every node/.style={draw,shape=circle,fill=red}]

\node[fill=black](px11){};
\node[below=1.4cm of px11, fill=black](px12){};
\node[below=1.4cm of px12, fill=black](px13){};
\node[below=1.4cm of px13, fill=black](px14){};
\node[below=1.4cm of px14, fill=black](px15){};
\node[below=1.4cm of px15, fill=black](px16){};
\node[below=1.4cm of px16, fill=black](px17){};
\node[below=1.4cm of px17, fill=black](px18){};
\node[below=1.4cm of px18, fill=black](px19){};
\node[below=1.4cm of px19, fill=black](px1a){};

\node[above left=0.27cm and 0.4cm of px11, fill=red](ex111){};
\node[left=0.29cm of px11, fill=red](ex112){};
\node[below left=0.27cm and 0.4cm of px11, fill=red](ex113){};
\node[right=0.29cm of px11, fill=orange](ex11a){};

\node[above left=0.27cm and 0.4cm of px12, fill=red](ex121){};
\node[left=0.29cm of px12, fill=red](ex122){};
\node[below left=0.27cm and 0.4cm of px12, fill=red](ex123){};
\node[right=0.29cm of px12, fill=orange](ex12a){};

\node[above left=0.27cm and 0.4cm of px13, fill=red](ex131){};
\node[left=0.29cm of px13, fill=red](ex132){};
\node[below left=0.27cm and 0.4cm of px13, fill=red](ex133){};
\node[right=0.29cm of px13, fill=orange](ex13a){};

\node[above left=0.27cm and 0.4cm of px14, fill=red](ex141){};
\node[left=0.29cm of px14, fill=red](ex142){};
\node[below left=0.27cm and 0.4cm of px14, fill=red](ex143){};
\node[right=0.29cm of px14, fill=orange](ex14a){};

\node[above left=0.27cm and 0.4cm of px15, fill=red](ex151){};
\node[left=0.29cm of px15, fill=red](ex152){};
\node[below left=0.27cm and 0.4cm of px15, fill=red](ex153){};
\node[right=0.29cm of px15, fill=orange](ex15a){};

\node[above left=0.27cm and 0.4cm of px16, fill=red](ex161){};
\node[left=0.29cm of px16, fill=red](ex162){};
\node[below left=0.27cm and 0.4cm of px16, fill=red](ex163){};
\node[right=0.29cm of px16, fill=orange](ex16a){};

\node[above left=0.27cm and 0.4cm of px17, fill=red](ex171){};
\node[left=0.29cm of px17, fill=red](ex172){};
\node[below left=0.27cm and 0.4cm of px17, fill=red](ex173){};
\node[right=0.29cm of px17, fill=orange](ex17a){};

\node[above left=0.27cm and 0.4cm of px18, fill=red](ex181){};
\node[left=0.29cm of px18, fill=red](ex182){};
\node[below left=0.27cm and 0.4cm of px18, fill=red](ex183){};
\node[right=0.29cm of px18, fill=orange](ex18a){};

\node[above left=0.27cm and 0.4cm of px19, fill=red](ex191){};
\node[left=0.29cm of px19, fill=red](ex192){};
\node[below left=0.27cm and 0.4cm of px19, fill=red](ex193){};
\node[right=0.29cm of px19, fill=orange](ex19a){};

\node[above left=0.27cm and 0.4cm of px1a, fill=red](ex1a1){};
\node[left=0.29cm of px1a, fill=red](ex1a2){};
\node[below left=0.27cm and 0.4cm of px1a, fill=red](ex1a3){};
\node[right=0.29cm of px1a, fill=orange](ex1aa){};

\node[right=1.93cm of px11, fill=black](px21){};
\node[below=1.4cm of px21, fill=black](px22){};
\node[below=1.4cm of px22, fill=black](px23){};
\node[below=1.4cm of px23, fill=black](px24){};
\node[below=1.4cm of px24, fill=black](px25){};
\node[below=1.4cm of px25, fill=black](px26){};
\node[below=1.4cm of px26, fill=black](px27){};
\node[below=1.4cm of px27, fill=black](px28){};
\node[below=1.4cm of px28, fill=black](px29){};
\node[below=1.4cm of px29, fill=black](px2a){};

\node[above right=0.27cm and 0.4cm of px21, fill=red](ex211){};
\node[right=0.29cm of px21, fill=red](ex212){};
\node[below right=0.27cm and 0.4cm of px21, fill=red](ex213){};
\node[left=0.29cm of px21, fill=orange](ex21a){};
\node[right=0.29cm of ex211, fill=blue](ex214){};
\node[right=0.29cm of ex212, fill=blue](ex215){};
\node[right=0.29cm of ex213, fill=blue](ex216){};
\node[right=0.29cm of ex214, fill=pink](ex217){};
\node[right=0.29cm of ex215, fill=pink](ex218){};
\node[right=0.29cm of ex216, fill=pink](ex219){};

\node[above right=0.27cm and 0.4cm of px22, fill=red](ex221){};
\node[right=0.29cm of px22, fill=red](ex222){};
\node[below right=0.27cm and 0.4cm of px22, fill=red](ex223){};
\node[left=0.29cm of px22, fill=orange](ex22a){};
\node[right=0.29cm of ex221, fill=blue](ex224){};
\node[right=0.29cm of ex222, fill=blue](ex225){};
\node[right=0.29cm of ex223, fill=blue](ex226){};
\node[right=0.29cm of ex224, fill=pink](ex227){};
\node[right=0.29cm of ex225, fill=pink](ex228){};
\node[right=0.29cm of ex226, fill=pink](ex229){};

\node[above right=0.27cm and 0.4cm of px23, fill=red](ex231){};
\node[right=0.29cm of px23, fill=red](ex232){};
\node[below right=0.27cm and 0.4cm of px23, fill=red](ex233){};
\node[left=0.29cm of px23, fill=orange](ex23a){};
\node[right=0.29cm of ex231, fill=blue](ex234){};
\node[right=0.29cm of ex232, fill=blue](ex235){};
\node[right=0.29cm of ex233, fill=blue](ex236){};
\node[right=0.29cm of ex234, fill=pink](ex237){};
\node[right=0.29cm of ex235, fill=pink](ex238){};
\node[right=0.29cm of ex236, fill=pink](ex239){};

\node[above right=0.27cm and 0.4cm of px24, fill=red](ex241){};
\node[right=0.29cm of px24, fill=red](ex242){};
\node[below right=0.27cm and 0.4cm of px24, fill=red](ex243){};
\node[left=0.29cm of px24, fill=orange](ex24a){};
\node[right=0.29cm of ex241, fill=blue](ex244){};
\node[right=0.29cm of ex242, fill=blue](ex245){};
\node[right=0.29cm of ex243, fill=blue](ex246){};
\node[right=0.29cm of ex244, fill=pink](ex247){};
\node[right=0.29cm of ex245, fill=pink](ex248){};
\node[right=0.29cm of ex246, fill=pink](ex249){};

\node[above right=0.27cm and 0.4cm of px25, fill=red](ex251){};
\node[right=0.29cm of px25, fill=red](ex252){};
\node[below right=0.27cm and 0.4cm of px25, fill=red](ex253){};
\node[left=0.29cm of px25, fill=orange](ex25a){};
\node[right=0.29cm of ex251, fill=blue](ex254){};
\node[right=0.29cm of ex252, fill=blue](ex255){};
\node[right=0.29cm of ex253, fill=blue](ex256){};
\node[right=0.29cm of ex254, fill=pink](ex257){};
\node[right=0.29cm of ex255, fill=pink](ex258){};
\node[right=0.29cm of ex256, fill=pink](ex259){};

\node[above right=0.27cm and 0.4cm of px26, fill=red](ex261){};
\node[right=0.29cm of px26, fill=red](ex262){};
\node[below right=0.27cm and 0.4cm of px26, fill=red](ex263){};
\node[left=0.29cm of px26, fill=orange](ex26a){};
\node[right=0.29cm of ex261, fill=blue](ex264){};
\node[right=0.29cm of ex264, fill=pink](ex265){};

\node[above right=0.27cm and 0.4cm of px27, fill=red](ex271){};
\node[right=0.29cm of px27, fill=red](ex272){};
\node[below right=0.27cm and 0.4cm of px27, fill=red](ex273){};
\node[left=0.29cm of px27, fill=orange](ex27a){};

\node[above right=0.27cm and 0.4cm of px28, fill=red](ex281){};
\node[right=0.29cm of px28, fill=red](ex282){};
\node[below right=0.27cm and 0.4cm of px28, fill=red](ex283){};
\node[left=0.29cm of px28, fill=orange](ex28a){};

\node[above right=0.27cm and 0.4cm of px29, fill=red](ex291){};
\node[right=0.29cm of px29, fill=red](ex292){};
\node[below right=0.27cm and 0.4cm of px29, fill=red](ex293){};
\node[left=0.29cm of px29, fill=orange](ex29a){};

\node[above right=0.27cm and 0.4cm of px2a, fill=red](ex2a1){};
\node[left=0.29cm of px27, fill=orange](ex27a){};
\node[right=0.29cm of px2a, fill=red](ex2a2){};
\node[below right=0.27cm and 0.4cm of px2a, fill=red](ex2a3){};
\node[left=0.29cm of px2a, fill=orange](ex2aa){};

\draw
(px11) -- (px12)
(px12) -- (px13)
(px13) -- (px14)
(px14) -- (px15)
(px15) -- (px16)
(px16) -- (px17)
(px17) -- (px18)
(px18) -- (px19)
(px19) -- (px1a)

(px11) -- (ex111)
(px11) -- (ex112)
(px11) -- (ex113)
(px11) -- (ex11a)

(px12) -- (ex121)
(px12) -- (ex122)
(px12) -- (ex123)
(px12) -- (ex12a)

(px13) -- (ex131)
(px13) -- (ex132)
(px13) -- (ex133)
(px13) -- (ex13a)

(px14) -- (ex141)
(px14) -- (ex142)
(px14) -- (ex143)
(px14) -- (ex14a)

(px15) -- (ex151)
(px15) -- (ex152)
(px15) -- (ex153)
(px15) -- (ex15a)

(px16) -- (ex161)
(px16) -- (ex162)
(px16) -- (ex163)
(px16) -- (ex16a)

(px17) -- (ex171)
(px17) -- (ex172)
(px17) -- (ex173)
(px17) -- (ex17a)

(px18) -- (ex181)
(px18) -- (ex182)
(px18) -- (ex183)
(px18) -- (ex18a)

(px19) -- (ex191)
(px19) -- (ex192)
(px19) -- (ex193)
(px19) -- (ex19a)

(px1a) -- (ex1a1)
(px1a) -- (ex1a2)
(px1a) -- (ex1a3)
(px1a) -- (ex1aa)

(px21) -- (px22)
(px22) -- (px23)
(px23) -- (px24)
(px24) -- (px25)
(px25) -- (px26)
(px26) -- (px27)
(px27) -- (px28)
(px28) -- (px29)
(px29) -- (px2a)

(px21) -- (ex211)
(px21) -- (ex212)
(px21) -- (ex213)
(px21) -- (ex21a)
(ex211) -- (ex214)
(ex212) -- (ex215)
(ex213) -- (ex216)
(ex214) -- (ex217)
(ex215) -- (ex218)
(ex216) -- (ex219)

(px22) -- (ex221)
(px22) -- (ex222)
(px22) -- (ex223)
(px22) -- (ex22a)
(ex221) -- (ex224)
(ex222) -- (ex225)
(ex223) -- (ex226)
(ex224) -- (ex227)
(ex225) -- (ex228)
(ex226) -- (ex229)

(px23) -- (ex231)
(px23) -- (ex232)
(px23) -- (ex233)
(px23) -- (ex23a)
(ex231) -- (ex234)
(ex232) -- (ex235)
(ex233) -- (ex236)
(ex234) -- (ex237)
(ex235) -- (ex238)
(ex236) -- (ex239)

(px24) -- (ex241)
(px24) -- (ex242)
(px24) -- (ex243)
(px24) -- (ex24a)
(ex241) -- (ex244)
(ex242) -- (ex245)
(ex243) -- (ex246)
(ex244) -- (ex247)
(ex245) -- (ex248)
(ex246) -- (ex249)

(px25) -- (ex251)
(px25) -- (ex252)
(px25) -- (ex253)
(px25) -- (ex25a)
(ex251) -- (ex254)
(ex252) -- (ex255)
(ex253) -- (ex256)
(ex254) -- (ex257)
(ex255) -- (ex258)
(ex256) -- (ex259)

(px26) -- (ex261)
(px26) -- (ex262)
(px26) -- (ex263)
(px26) -- (ex26a)
(ex261) -- (ex264)
(ex264) -- (ex265)

(px27) -- (ex271)
(px27) -- (ex272)
(px27) -- (ex273)
(px27) -- (ex27a)

(px28) -- (ex281)
(px28) -- (ex282)
(px28) -- (ex283)
(px28) -- (ex28a)

(px29) -- (ex291)
(px29) -- (ex292)
(px29) -- (ex293)
(px29) -- (ex29a)

(px2a) -- (ex2a1)
(px2a) -- (ex2a2)
(px2a) -- (ex2a3)
(px2a) -- (ex2aa)

(ex1a3) -- (ex2a3)
;

\end{tikzpicture} 		}
		\caption{$T_2$}
	\end{subfigure}
	\caption{Example of the construction in the proof of \refthm{NPhardBoundedTree} where $m$=2, $A$=10, $a_1$=$a_2$=4 and $a_3$=$a_4$=$a_5$=$a_6$=3. Best viewed in color.}
	\label{fig:BoundedTrees}
\end{figure*}

\begin{proof}
	The proof is done by reduction from the \NP-hard \threepartition problem \cite{Intractability} that is defined as follows.
	Given the integers $A$ and $a_1,\hdots,a_{3m}$ in unary representation, such that $\sum_{i=1}^{3m}a_i=mA$ and $A/4<a_i<A/2$ for $1\leq i\leq 3m$, decide whether there exists a partition of $a_1,\hdots,a_{3m}$ into $m$ groups of size 3 such that the elements in each group sum up to exactly $A$.
	For technical reasons, we restrict the reduction to $A\geq8$. However, the ignored cases are trivial. Precisely, for $A\in\{6,7\}$ the answer is always YES and for $A\in[5]$ there exist no valid instances.
	
	Given an instance of \threepartition with $A\geq8$, we construct two trees $T_1$ and $T_2$ such that the given instance has answer YES, if and only if $\distp(T_1,T_2)\leq\boundp{2}$. \reffig{BoundedTrees} shows an example of $T_1$ and $T_2$ for $m=2$ and $A=10$.
	For illustrative reasons we assign each node a color during the construction.
	The colors are used in the example and we will refer to certain nodes by their color later in the proof.
	However, they do not restrict the possible alignments in any way.
	
	Initialize $T_1$ as the disjoint union of paths $p^1_1,\hdots,p^1_{3m}$ such that $p^1_i$ has $a_i$ black nodes; initialize $T_2$ as the disjoint union of paths $p^2_1,\hdots,p^2_{m}$ consisting of $A$ black nodes each.
	In the following we refer to one endpoint of $p^k_i$ as $e_1(p^k_i)$ and the other endpoint as $e_2(p^k_i)$.
	We attach three new red leaves and one new orange leaf to each black node in both $T_1$ and $T_2$.	
	Next we modify the graphs into trees by connecting the paths.
	For $1\leq i\leq3m-1$ we add an edge between the orange leaf adjacent to $e_1(p^1_i)$ and the orange leaf adjacent to $e_2(p^1_{i+1})$.
	For $1\leq i\leq m-1$ we add an edge between one of the red leaves adjacent to $e_1(p^2_i)$ and one of the red leaves adjacent to $e_2(p^2_{i+1})$.
	
	Next we attach two new pink leaves to each inner (non-endpoint) path node in $T_1$ and attach a new blue leaf to each pink node.
	Then we add the same number of blue nodes to $T_2$ and connect each blue node to one of the red nodes with degree 1.
	Finally, we attach a new pink leaf to each blue node.
	Note that both $T_1$ and $T_2$ are trees with bounded degree.
	Precisely, the highest degree in $T_1$ is 8 and 6 in $T_2$ independent of the problem instance.
	
	Intuitively, the construction ensures that every inner path node has already 2 mismatches just because of the degree difference.
	If there is no partition, at least one path in $T_1$ cannot be mapped contiguously into a path in $T_2$ which raises the \mc of some inner path node to at least 3.	
	Simultaneously, the construction ensures that there is an alignment for which the mismatch graph consists only of stars with maximum degree 2 if there is an alignment.
	
	The formal continuation of this proof can be found in \refapx{ComplexityOperatorNorms}.
\end{proof}

\section{Approximability for Operator Norms}
\label{sec:approx}
In this section we investigate the approximability of \Distp and \Dist{|p|} for $1\leq p\leq\infty$.
Again, we omit specifying the possible values for $p$ and mentioning \Dist{|p|} explicitly as the proofs work the same for it.
We consider the following two possibilities to measure the error of an approximation algorithm for a minimization problem.
An algorithm $\mathcal{A}$ has \emph{multiplicative error} $\alpha>1$, if for any instance $\mathcal{I}$ of the problem with an optimum $\text{OPT}(\mathcal{I})$, $\mathcal{A}$ outputs a solution with value $\mathcal{A}(\mathcal{I})$ such that $\text{OPT}(\mathcal{I})\leq\mathcal{A}(\mathcal{I})\leq\alpha\text{OPT}(\mathcal{I})$.
In this case we call $\mathcal{A}$ an \emph{$\alpha$-approximation algorithm}.
An algorithm $\mathcal{B}$ has \emph{additive error} $\varepsilon>0$, if $\text{OPT}(\mathcal{I})\leq\mathcal{B}(\mathcal{I})\leq\text{OPT}(\mathcal{I})+\varepsilon$ for any instance $\mathcal{I}$.

Approximating \Distp with multiplicative error is at least as hard as the graph isomorphism problem (GI) since such an approximation algorithm $\mathcal{A}$ could be used to decide GI considering that $\mathcal{A}(G,H)=0$ if and only if $G$ is isomorphic to $H$.

Furthermore, we can deduce thresholds under which the $\alpha$-approximation is NP-hard using the gap between the decision values of the reduction in each hardness proof from \refsec{ComplexityOperatorNorms}.
\begin{theorem}
	\label{thm:MuliplicativeApproxThreshold}
	For $1\leq p\leq\infty$ and any $\varepsilon>0$, unless $\PTIME=\NP$, there is no polynomial time approximation algorithm for \Distp or \Dist{|p|} with a multiplicative error guarantee of
	\begin{enumerate}
		\item $\boundp{3}-\varepsilon$, even if both input graphs have bounded degree,
		\item $\boundp{2}-\varepsilon$, even if one input graph is a path and the other one has bounded degree,
		\item $\frac{\boundp{3}}{\boundp{2}}-\varepsilon$, even if both input graphs are trees with bounded degree.
	\end{enumerate}
	\begin{proof}
		We recall the proof of \refthm{NPhard}.
		If $G$ has a Hamiltonian cycle, then $\distp(C_n,G)\leq\boundp{1}=1$.
		Otherwise $\distp(C_n,G)\geq\boundp{3}$.
		Assume there is a polynomial time approximation algorithm $\mathcal{A}$ with a multiplicative error guarantee of $\boundp{3}-\varepsilon$ for $\varepsilon>0$.
		Then we	can distinguish the two cases by checking whether $\mathcal{A}(C_n,G)<\boundp{3}$ and therefore decide \hamcycle on 3-regular graphs in polynomial time.
		The same argument can be used for the other bounds using the proofs of \refthm{NPhardPath} and \refthm{NPhardBoundedTree}, respectively.
	\end{proof}
\end{theorem}

In particular, this implies that there is no polynomial time approximation scheme (PTAS) for \Distp or \Dist{|p|} under the respective restrictions.

Next we show the additive approximation hardness by scaling up the gap between the two decision values of the reduction in the proof of \refthm{NPhard}.
For this we replace each edge with a colored gadget and then modify the graph so that an optimal alignment has to be color-preserving.
The proof of \refthm{AdditiveApproxHardness} can be found in \refapx{approx}.
\begin{theorem}
	\label{thm:AdditiveApproxHardness}
	For $1\leq p\leq\infty$ there is no polynomial time approximation algorithm for \Distp with any constant additive error guarantee unless $\PTIME=\NP$.
\end{theorem}

However, the approximation of \Distp becomes trivial once we restrict the input to graphs of bounded degree, although \Distp stays \NP-hard under this restriction.
The proof of \refthm{ApproxAlgorithm} can be found in \refapx{approx}.
\begin{theorem}
	\label{thm:ApproxAlgorithm}
	For $1\leq p\leq\infty$, if one graph has maximum degree $d$, then there is a polynomial time approximation algorithm for \Distp and \Dist{|p|} with
	\begin{enumerate}
		\item a constant additive error guarantee $2d$,
		\item a constant multiplicative error guarantee $1+2d$.
	\end{enumerate}
\end{theorem}

\section{Complexity for Cut Norm}
\label{sec:CutNorm}
Finally, we show the hardness for \Dist{\square} which corresponds to the cut distance $\hat{ \delta}_\square$ (see
\cite[Chapter~8]{lov12}.
For any signed graph $G$ and $V\subseteq V(G)$ the induced subgraph $G[V]$ is the signed graph with vertex set $V$, $E_+(G[V])=\{vw\in E_+(G)\mid v,w\in V)\}$, and $E_-(G[V])=\{vw\in E_-(G)\mid v,w\in V)\}$.
\begin{lemma}
	\label{lem:CutNormInducedSubgraph}
	Let $\Delta$ be a signed graph and $W\subseteq V(\Delta)$.
	Then $\mu_\square(\Delta)\geq\mu_\square(\Delta[W])$.
	\begin{proof}
		Let $V\coloneqq V(\Delta)$, $A\coloneqq A_\Delta$, $B\coloneqq A_{\Delta[W]}$ and $S',T'\coloneqq\argmax_{S,T\subseteq
			W}\left|\sum_{v\in S,w\in T}B_{vw}\right|$.
		Then
		\begin{equation*}			
			\norm{\Delta[W]}_\square
			=\left|\sum_{v\in S',w\in T'}B_{vw}\right|
			=\left|\sum_{v\in S',w\in T'}A_{vw}\right|
			\leq\max_{S,T\subseteq V}\left|\sum_{v\in S,w\in T}A_{vw}\right|
			=\norm{\Delta}_\square.
			\qedhere
		\end{equation*}
\end{proof}
\end{lemma}

Our hardness proof for \Dist{\square} is based on the hardness of computing the cut norm.
With a construction similar to the one in \reflem{ColorConversion} we enforce a specific alignment.
The proof can be found in \refapx{cut}.
	
\begin{theorem}\label{theo:cut}
	The problem \Dist{\square} is \NP-hard.
\end{theorem}

\section{Concluding Remarks}
We study the computational complexity of a class of graph metrics
based on mismatch norms, or equivalently, matrix norms applied to the
difference of the adjacency matrices of the input graphs under an
optimal alignment between the vertex sets. We find that computing the
distance between graphs under these metrics (at least for the
standard, natural
matrix norms)  is \NP-hard, often already on simple input
graphs such as trees. This was essentially known for entrywise matrix
norms. We prove it for operator norms and also for the cut norm.

We leave it open to find (natural) general conditions on a mismatch
norm such that the corresponding distance problem becomes hard. Maybe
more importantly, we leave it open to find meaningful tractable
relaxations of the distance measures.

Measuring similarity via mismatch norms is only one approach. There
are several other, fundamentally different ways to measure
similarity. We are convinced that graph similarity deserves a
systematic and general theory that compares the different approaches
and studies their semantic as well as algorithmic properties. Our
paper is one contribution to such a theory.

\newpage
\appendix

\section{Appendix}
\subsection{Missing Proofs from Section \ref{sec:ComplexityOperatorNorms}}
\label{apx:ComplexityOperatorNorms}
\begin{proof}[Proof of \reflem{MMCupper2}]
	By definition we have $\mu_1(G^\pi-H)=\norm{A^\pi_G-A_H}_1$.
	It is well known that the $\ell_1$-operator norm boils down to the maximum absolute column sum:
	\[
	\norm{A}_1
	=\max_{\norm{x}_1=1}\norm{Ax}_1
	=\max_{j\in[n]}\sum^m_{i=1}\vert A_{ij}\vert.
	\]
	The $v$-th column of the matrix $A^\pi_G-A_H$ has an entry with value $+1$ or $-1$ for each mismatch of the node $v$ and all its other entries are 0.
	Therefore, the \mc of $v$ is equal to the absolute sum of all entries in the $v$-th column.
	It follows that the \mmc is equal to the maximum absolute column sum.
	Similarly, the $\ell_\infty$-operator norm boils down to the maximum absolute row sum:
	\[
	\norm{A}_\infty
	=\max_{\norm{x}_\infty=1}\norm{Ax}_\infty
	=\max_{i\in[m]}\sum^n_{j=1}\vert A_{ij}\vert.
	\]
	
	Due to the symmetry of the matrix $A^\pi_G-A_H$ we can apply the same argument for $\mu_\infty$.		
	Since all $\ell_p$-operator norms can be estimated in terms of the $\ell_1$-operator norm and $\ell_\infty$-operator norm \cites{EstimatingTheMatrixPNorm}[page 29]{PerturbationTheoryForLinearOperators}\ by
	\[
	\norm{A}_p\leq\norm{A}_1^{1/p}\cdot\norm{A}_\infty^{1-1/p},
	\]
	we get the upper bound
	\begin{align*}
		\mu_p(G^\pi-H)
		=&\norm{A^\pi_G-A_H}_p\\
		\leq&\norm{A^\pi_G-A_H}_1^{1/p}\cdot\norm{A^\pi_G-A_H}_\infty^{1-1/p}\\
		=&\fmmc(\pi)^{1/p}\cdot \fmmc(\pi)^{1-1/p}\\
		=&\fmmc(\pi).
		\qedhere
	\end{align*}
\end{proof}

\begin{proof}[Proof of \reflem{MMR}]
	Let $M$ be the adjacency matrix of $G^\pi-H$ and $\mathcal{S}$ be the set of submatrices of $M$ corresponding to the connected components in $\mathcal{C}$.
	We show $\norm{M}_p=\max_{S\in\mathcal{S}}\norm{S}_p$ by induction over $n\coloneqq\abs{\mathcal{S}}$.	
	For $n=1$ we have $\mathcal{S}=\{S\}$ with $S=M$.	
	For the induction step ($n\to n+1$) let $S_1$ be the $(n+1)$-th submatrix of $M$ and $S_2$ the smallest submatrix of $M$ that contains the first $n$ submatrices of $M$.
	We can permute $M$ to the form
	\[
	M^\tau=\begin{pmatrix}
		S_1 & 0\\
		0 & S_2\\
	\end{pmatrix}.
	\]
	Due to the permutation invariance of $\ell_p$-operator norms it holds that $\norm{M}_p=\norm{M^\tau}_p$.
	Now let $x_1\coloneqq\argmax_{\norm{x}_p=1} \norm{S_1x}_p$ and $x_2\coloneqq\argmax_{\norm{x}_p=1} \norm{S_2x}_p$.
	Then
	\begin{align*}
		\norm{M}_p
		=&\max_{\norm{x}_p=1} \norm{Mx}_p\\
		=&\max_{\vert\alpha_1\vert^p+\vert\alpha_2\vert^p=1}\norm{
			\begin{pmatrix}
				S_1 \cdot \alpha_1 x_1\\
				S_2 \cdot \alpha_1 x_2
		\end{pmatrix}}_p\\
		=&\max_{\vert\alpha_1\vert^p+\vert\alpha_2\vert^p=1}
		\left(\norm{\alpha_1 S_1}^p_p + \norm{\alpha_2 S_2}^p_p\right)^{1/p}\\
		=&\max_{\vert\alpha_1\vert^p+\vert\alpha_2\vert^p=1}
		\left(\vert\alpha_1\vert^p\norm{S_1}^p_p + \vert\alpha_2\vert^p\norm{S_2}^p_p\right)^{1/p}\\	
		=&\max_{\alpha\in[0,1]}
		\left(\alpha^p\norm{S_1}^p_p + (1-\alpha^p)\norm{S_2}^p_p\right)^{1/p}\\	
		=&\max_{\alpha\in[0,1]}
		\left(\alpha^p\left(\norm{S_1}^p_p-\norm{S_2}^p_p\right)\right)^{1/p}.
	\end{align*}
	Now we compute $\alpha$ as
	\begin{align*}
		\argmax_{\alpha\in[0,1]}
		\left(\alpha^p\left(\norm{S_1}^p_p-\norm{S_2}^p_p\right)\right)^{1/p}
		=&\argmax_{\alpha\in[0,1]}
		\left(\alpha^p\left(\norm{S_1}^p_p-\norm{S_2}^p_p\right)\right)\\
		=&\argmax_{\alpha\in[0,1]}
		\left(\alpha\left(\norm{S_1}^p_p-\norm{S_2}^p_p\right)\right)\\
		=&
		\begin{cases}
			1, &\text{ if }\norm{S_1}^p_p \geq \norm{S_2}^p_p,\\
			0, &\text{ else,}
		\end{cases}\\
		=&
		\begin{cases}
			1, &\text{ if }\norm{S_1}_p \geq \norm{S_2}_p,\\
			0, &\text{ else.}
		\end{cases}
	\end{align*}
	Hence, $\norm{M}_p=\max\left(\norm{S_1}_p,\norm{S_2}_p\right)$, which concludes the induction step.
\end{proof}

\begin{proof}[Proof of \reflem{MCG}]	
	Let $M$ be the adjacency matrix of $G^\pi-H$.
	First we show $\norm{M}_p\geq c^{1/p}$. Let $x\in\R^n$ with $x_v=1$ and $x_i=0$ for $i\neq v$, then:
	\[
	\norm{M}_p
	\geq\norm{Mx}_p
	=\norm{M_{*v}}_p
	=\left(\vert M_{1v}\vert^p+\hdots+\vert M_{nv}\vert^p\right)^{1/p}
	=\big(\underbrace{1^p+\hdots+1^p}_c\big)^{1/p}
	=c^{1/p}.
	\]
	In the second to last step we used the fact that there are exactly $c$ entries equal to 1 or $-1$ in the $v$-th column of $M$ while all other entries in that column are 0.
	Next, we define $x'\in\R^n$ as
	\[
	x'_i=
	\begin{cases}
		c^{-1/p}, &\text{ if } M_{vi}=1\text{ ($i$ is a mismatch of $v$)}, \\
		0, &\text{ else.} \\
	\end{cases}
	\]
	Note that $\norm{x'}_p=1$.
	We use $x'$ to show the other lower bound:
	\begin{align*}
		\norm{M}_p
		\geq&\norm{Mx'}_p
		=\norm{\begin{pmatrix}
				\sum_i M_{1i} x'_i\\
				\vdots\\
				\sum_i M_{ni} x'_i
		\end{pmatrix}}_p
		=\left(\bigg\vert \sum_i M_{1i} x'_i\bigg\vert^p+\dots+\bigg\vert\sum_i M_{ni} x'_i\bigg\vert^p\right)^{1/p}\\
		\geq&\left(\bigg\vert \sum_i M_{vi} x'_i\bigg\vert^p\right)^{1/p}
		=\bigg\vert \sum_i M_{vi} x'_i\bigg\vert
		=\bigg\vert \sum_{\substack{i\text{ with }\\\vert M_{vi}\vert=1}} M_{vi}c^{-1/p}\bigg\vert\\
		=&\bigg\vert c^{-1/p}\sum_{\substack{i\text{ with }\\\vert M_{vi}\vert=1}} M_{vi} \bigg\vert
		=\bigg\vert c^{-1/p}\sum_{\substack{i\text{ with }\\\vert M_{vi}\vert=1}} 1 \bigg\vert
		=\vert c^{-1/p}c\vert
		=c^{1-1/p}.
	\end{align*}

	Now, we assume $G^\pi-H$ is a star and show $\mu_p(G^\pi-H)=\boundp{\fmmc(\pi)}$.
	$M$ has exactly $c$ entries with value 1 or -1 in the column and the row corresponding to $v$, which represent the $c$ mismatches; all other entries are 0.
	Hence, we can permute $M$ to the form
	\[
	M^\tau=\begin{pmatrix}
		0 & 1 & \hdots & 1 & -1 & \hdots & -1\\
		1 & & & & & &\\
		\vdots & & & & & &\\
		1 & & & \makebox(0,0){\vspace{-0.2cm}\hspace{0.7cm}\text{\huge0}} & & &\\
		-1 & & & & & &\\
		\vdots & & & & & &\\		
		-1 & & & & & &
	\end{pmatrix}
	\in\R^{(c+1)\times(c+1)}.
	\]
Due to the permutation invariance of $\ell_p$-operator norms it holds that $\norm{M}_p=\norm{M^\tau}_p$.
	Using the same argument as in the proof of \reflem{MMR} we get
	\[
	\norm{M}_p=\max\left(\norm{V^T}_p,\norm{V}_p\right),\ \ \ 
	V\coloneqq(M_{1,2},\hdots,M_{1,c+1})=(1,\hdots,1,-1,\hdots,-1)\in\R^{1\times c}.
	\]
	The computation of the first maximization term is straightforward because $x$ is one-dimensional:
	\[
	\underbrace{\norm{V^T}_p}_{\text{matrix norm}}
	=\max_{\norm{x}_p=1}\underbrace{\norm{V^Tx}_p}_{\text{vector norm}}
	=\underbrace{\norm{V^T(1)}_p}_{\text{vector norm}}
	=\underbrace{\norm{V^T}_p}_{\text{vector norm}}
	=\big(\underbrace{1^p+\hdots+1^p}_c\big)^{1/p}
	=c^{1/p}.
	\]
	The optimal solution for the other maximization term is exactly $x'$:
	\begin{align*}
		\argmax_{\norm{x}_p=1}\norm{Vx}_p
		=\;&\argmax_{\norm{x}_p=1}\sqrt[p]{(x_1+\hdots+x_c)^p}\\
		=\;&\argmax_{\norm{x}_p=1}x_1+\hdots+x_c\\
		=\;&\argmax_{\abs{x_1}^p+\cdots+\abs{x_c}^p=1}x_1+\hdots+x_c\\
		=\;&x'.
	\end{align*}
	Therefore, we know that $\norm{V}_p=\norm{Vx'}_p=c^{1-1/p}$.
\end{proof}

\begin{proof}[Proof of \reflem{MMCDistinguish}]
	If $p\leq2$ then
	\[
	\boundp{c}=\max\left(c^{1-1/p},c^{1/p}\right)=c^{1/p}<d^{1/p}=\max\left(d^{1-1/p},d^{1/p}\right)=\boundp{d}.
	\]
	Otherwise
	\[
	\boundp{c}=\max\left(c^{1-1/p},c^{1/p}\right)=c^{1-1/p}<d^{1-1/p}=\max\left(d^{1-1/p},d^{1/p}\right)=\boundp{d}.	
	\qedhere
	\]
\end{proof}

\begin{proof}[Proof of \refthm{NPhardPath}]
	We modify the proof of \refthm{NPhard} so that the comparison graph is a path instead of the cycle $C_n$.		
	Given a 3-regular graph $G$ of order $n$ as an instance of \hamcycle, let $P_{n+2}$ be the path of length $n+2$ and $v$ be a node randomly chosen from $G$.
	We remove one of the edges incident to $v$ and attach a new leaf to each of the two separated nodes.
	Now we create another copy of $G$ with the same modification applied to a different edge incident to $v$ and call the new graphs $G_1$ and $G_2$.
	One of the two has a Hamiltonian path if and only if $G$ has a Hamiltonian cycle.
	
	The proof follows from the following claim: $G_i$ has a Hamiltonian path if and only if $\distp(P_{n+2},G_i)\leq\boundp{1}$.		
	Indeed, if $G_i$ has a Hamiltonian path, the natural bijection $\pi:V(P_{n+2})\rightarrow V(G_i)$, which aligns the path $P_{n+2}$ with the Hamiltonian path in $G_i$, will cause exactly an \mc of 1 on each original node (nodes from $G$) and an \mc of 0 on the two newly inserted leaves.
	
	Conversely, assume $\distp(P_{n+2},G_i)\leq\boundp{1}$.
	Again, this implies every node has at most one mismatch.
	Then, the path end nodes in $P_{n+2}$ can only be mapped to the new nodes in $G_i$ because the degree of the original nodes in $G_i$ differs by 2.
	Therefore, inner path nodes are mapped to nodes from $G$.
	Since one of their incident edges in $G$ is always mismatched due to the degree difference, the other two incident edges have to be matched correctly.
	Altogether, every edge in $P_{n+2}$ is matched correctly, which is a contradiction to the assumption.
\end{proof}

\begin{proof}[Continuation of the proof of \refthm{NPhardBoundedTree}]	
	We have to show that the \threepartition instance has answer YES if and only if $\distp(T_1,T_2)\leq\boundp{2}$.
	First we show via contradiction that $\distp(T_1,T_2)\geq\boundp{3}$ if the \threepartition instance has answer NO.
	According to \reflem{MMCDistinguish} this implies $\distp(T_1,T_2)>\boundp{2}$.
	So, we assume that the \threepartition instance has answer NO and there exists an alignment $\pi$ from $T_1$ to $T_2$ such that $\mu_p(T_1^\pi-T_2)<\boundp{3}$.
	According to \reflem{MMClowerG} the \mmc of $\pi$ has to be smaller than 3.
	In particular, nodes with a degree $d$ can only be mapped to nodes with a degree between $d-2$ and $d+2$.
	
	Inner path nodes in $T_1$ have degree 8 and can only be mapped to inner path nodes in $T_2$, since there are no other nodes in $T_2$ that have a degree of at least $8-2=6$.
	Using the analog degree argument, path endpoints in $T_1$ can only be mapped to path nodes in $T_2$. However, $\pi$ needs to be a bijection and the inner path nodes are already matched, so path endpoints can only be mapped to path endpoints.	
	All in all, we have shown that under our assumption that the \mmc of $\pi$ is smaller than 3, black nodes can only be mapped to black nodes.
	
	The degree of the inner path nodes in $T_1$ is at least by 2 higher than the degree of any path node in $T_2$.
	Hence, each inner path node has already two mismatches from $T_1$ to $T_2$.
	If two adjacent path nodes in $T_1$ would be mapped to two non-adjacent path nodes in $T_2$, at least one of them is an inner path node and would then have a third mismatch, so the \mmc would be at least 3.
	Intuitively, this means the internal structure of the paths cannot be destroyed by the alignment.
	Formally, for each path $p^2_i$ in $T_2$ there have to be three paths $p^1_{i_1},p^1_{i_2},p^1_{i_3}$ in $T_1$ such that the subgraph of $T_2$ induced by $\pi\left(p^1_{i_1} \cup p^1_{i_2} \cup p^1_{i_3}\right)$ is equal to $p^2_i$.
	It follows that $\abs{p_i^2}=\abs{p_{i_1}^1}+\abs{p_{i_2}^1}+\abs{p_{i_3}^1}$, so $A=a_{i_1}+a_{i_2}+a_{i_3}$ for $i\in[n]$.
	This corresponds exactly to the existence of a partition satisfying the requirements of \threepartition, which is a contradiction to our assumption.
	
	Note that our argument relies on the property $A/4<a_i<A/2$ for $1\leq i\leq 3m$ of the \threepartition instance because otherwise more or less than three paths in $T_1$ might be packed into one of the paths in $T_2$. Also, instances with $a_i=2$ would be allowed without our restriction to $A\geq8$. In such a case the corresponding path $i$ would not have an inner path node.
	
	Now we assume the desired partition exists and show that we can find an alignment $\pi$ with $\mu_p(T_1^\pi-T_2)\leq\boundp{2}$.
	Note that our choice of $\pi$ will map nodes in $T_1$ only to nodes that have the same color in $T_2$.
	For each group in the partition we pack the three corresponding paths in $T_1$ into one of the paths in $T_2$.
	Formally, let $\{j,k,l\}$ be the $i$-th group and $p^1_j$, $p^1_k$, $p^1_l$ the corresponding paths in $T_1$.
	Then the subgraphs of $T_2$ induced by $\pi(p^1_j)$, $\pi(p^1_k)$, $\pi(p^1_l)$, and $\pi(p^1_j \cup p^1_k \cup p^1_l)$ are paths consisting only of black nodes.
	
	Each blue node in $T_1$ is mapped to any of the blue nodes in $T_2$. By definition $\pi$ is required to be a bijection. It is possible to find such a $\pi$ because the number of blue nodes in $T_1$ is the same as in $T_2$.
	Each pink node in $T_1$ is mapped to the pink node in $T_2$ that has the ``same'' blue neighbor.
	Formally, let $p^1$ be a pink node in $T_1$ and $b^1$ its blue neighbor.
	Let $p^2$ be the the pink node adjacent to $\pi(b^1)$ in $T_2$. Then $\pi(p^1)=p^2$ holds.
	
	Each red node in $T_1$ is mapped to a red node in $T_2$ that has the ``same'' black neighbor. Formally, let $r^1_1$, $r^1_2$, $r^1_3$ be the red nodes adjacent to a path node $b^1$ in $T_1$. Let $r^2_1$, $r^2_2$, $r^2_3$ be the red nodes adjacent to $\pi(b^1)$ in $T_2$. Then $\pi(r^1_1)=r^2_1$, $\pi(r^1_2)=r^2_2$, and $\pi(r^1_3)=r^2_3$ holds.
	Each orange node in $T_1$ is mapped to the orange node in $T_2$ that has the ``same'' black neighbor.
	
	According to \reflem{MMCupperG}, it suffices to show that the \mmc of $\pi$ is 2 and that the mismatched neighbors of a node with an \mc of 2 have an \mc of 1.
	The second condition is called \emph{star condition} in the following.
	
	Blue nodes have degree 1 in $T_1$ and degree 2 in $T_2$.
	Since their pink neighbor is matched correctly by definition of $\pi$, we know that their \mc is exactly 1.
	Pink nodes have degree 2 in $T_1$ and degree 1 in $T_2$.
	Since their blue neighbor is matched correctly, their \mc is exactly 1.
	Red nodes are mapped next to the same black neighbor.
	Some red nodes have an additional edge to a blue node or to another red node in $T_2$.
	Hence, their \mc can be 0 or 1.
	Orange nodes are mapped to orange nodes with the same black neighbor.
	However, some orange nodes have an additional edge to another orange node in $T_1$.
	So their \mc can be 0 or 1.
	
	Path endpoints have no mismatches. Inner path nodes have an \mc of 2.
	Since their mismatches, which are pink nodes, have an \mc of at most 1, the star condition is fulfilled.
	
	Using \reflem{MMCupperG} we finally conclude that $\distp(T_1,T_2)\leq\boundp{2}$ if the \threepartition instance has answer YES and therefore prove the theorem.
\end{proof}

\subsection{Missing Proofs from Section \ref{sec:approx}}
For the proof of \refthm{AdditiveApproxHardness} we use node coloring which we introduce in the following definition.
\label{apx:approx}
\begin{definition}
	\label{def:Color}
	Let $G$ and $H$ be two colored graphs of order $n$ with the respective colorings $c_G:V(G)\to\N$ and $c_H:V(H)\to\N$ such that $G$ and $H$ have the same number of nodes of each color. The problem \CDistp is the task to compute the following distance function for a given $\ell_p$-operator norm:
	\begin{align*}
		\cdistp(G,H)\coloneqq&\min_{\pi\in\Pi(G,H)}\norm{A_G^\pi-A_H}_p,\\
		\Pi(G,H)\coloneqq&\;\{\pi\in S_n\,|\,c_G(v)=c_H(v^\pi)\text{ for all }v\in V(G)\}.
	\end{align*}
\end{definition}

The following lemma allows us to color nodes if a polynomial increase of the input graph size is acceptable.
\begin{lemma}
	\label{lem:ColorConversion}
	Every instance $(G,H)$ of \CDistp can be converted to an instance $(G',H')$ of \Distp such that $\cdistp(G,H)=\distp(G',H')$ and the order of $G'$ is polynomial in the order of $G$.
\begin{proof}
	First, we compute an upper bound for $\cdistp(G,H)$.
	Let $n$ denote the order of the input graphs.
	Obviously, the \mmc of any alignment $\pi$ from $G$ to $H$ cannot be greater than $n-1$.
	According to \reflem{MMCupper2} it follows that $\cdistp(G,H)\leq\mu_p(G^\pi-H)\leq n-1$.
	
	Now we want to modify $G$ and $H$ such that any alignment, that does not preserve the coloring, has a mismatch norm of at least $n-1$, which is at least as high as the mismatch norm of any color-preserving alignment.
	Then we can drop the coloring without affecting the graph distance.
	Let $c$ be the number of colors that are used for $G$ and $H$ and let the colors be represented by the integers $1,\hdots,c$.
	For $1\leq i\leq c$ we attach $i\cdot n^2$ new leaves to each node in $G$ and $H$ that has the color $i$.
	Then we drop the coloring and denote the resulting graphs as $G_2$ and $H_2$.
	
	Any alignment $\pi$ from $G$ to $H$ can be converted to an alignment $\tau$ from $G_2$ to $H_2$ with $\mu_p(G^\pi-H)=\mu_p(G^\tau_2-H_2)$ as follows.
	For each non-leaf node $v$ in $G_2$ we set $v^\tau=v^\pi$ assuming that the original nodes of $G$ and $H$ are not renamed within the modification to $G_2$ and $H_2$.
	Each leaf in $G_2$ is mapped adjacency-preserving.
	Precisely, let $w$ be the neighbor of a leaf $v$ in $G_2$, then $v$ is mapped to a leaf in $H_2$ that is adjacent to $w^\tau$.
	Hence, there are no additional mismatches due to the leaves in $G_2$ and so $\mu_p(G^\pi-H)=\mu_p(G^\tau_2-H_2)$.
	
	Now assume there exists an alignment $\tau$ from $G_2$ to $H_2$ that does not preserve the original coloring and has $\mu_p(G^\tau_2-H_2)<n-1$.
	Each non-leaf node in $G_2$ has degree $n^2$ or more, while each leaf in $H_2$ has degree 1.
	Hence, if $\tau$ maps a non-leaf node to a leaf, then $\fmmc(\tau)\geq n^2-1$ and according to \reflem{MMClowerG} it follows that $\mu_p(G^\tau_2-H_2)\geq\boundp{n^2-1}$.
	However, $\boundp{n^2-1}=\max\left((n^2-1)^{1/p},(n^2-1)^{1-1/p}\right)\geq\sqrt{n^2-1}\geq n-1$ for every $\ell_p$-operator norm, which is a contradiction to our assumption.
	So, $\tau$ can only map leaves to leaves and non-leaves to non-leaves.
	
	If $\tau$ does not preserve the original coloring, then it also maps some node $v$ in $G_2$ to a node $w$ in $H_2$ such that the color of $v$ in $G$ is greater than the color of $w$ in $H$.
	This means that $v$ has at least $n^2$ more leaves in $G_2$ than $w$ in $H_2$.
	We get $\fmmc(\tau)\geq n^2$ and following the same argument as before this is a contradiction to our assumption.
	All in all, we conclude that the value of $\distp(G_2,H_2)$ is determined by an alignment that does preserve the original coloring, so $\distp(G_2,H_2)=\cdistp(G,H)$.
	Note that the order of $G_2$ and $H_2$ is polynomial in the order of $G$ because we add at most $c\cdot n^2\leq n^3$ leaves to each node in $G$.
\end{proof}
\end{lemma}

\begin{figure*}[htb]
	\centering
	\begin{tikzpicture}[
		every node/.style={draw,shape=circle}]
		\node[fill=black](e1){};
		\node[right=2cm of e1, fill=blue](b21){};
		\node[below=0.5cm of b21, fill=red](r21){};
		\node[left=0.5cm of r21, fill=red](r22){};
		\node[right=0.5cm of r21, fill=red](r23){};
		\node[right=0.5cm of r23, fill=red](r24){};
		\node[right=0.5cm of r24, fill=red](r25){};
		\node[above=0.5cm of b21, fill=red](r11){};
		\node[left=0.5cm of r11, fill=red](r12){};
		\node[right=0.5cm of r11, fill=red](r13){};
		\node[right=0.5cm of r13, fill=red](r14){};
		\node[right=0.5cm of r14, fill=red](r15){};
		\node[above=0.5cm of r11, fill=blue](b11){};
		\node[below=0.5cm of r21, fill=blue](b31){};
		\node[above=0.5cm of r14, fill=blue](b12){};
		\node[below=0.5cm of r14, fill=blue](b22){};
		\node[below=0.5cm of r24, fill=blue](b32){};
		\node[right=2cm of b22, fill=black](e2){};
		\draw		
		(e1)--(b21)
		(e1)|-(b11)
		(e1)|-(b31)
		(e2)--(b22)
		(e2)|-(b12)
		(e2)|-(b32)
		
		(b11)--(r11)
		(b11)--(r12)
		(b11)--(r13)
		(b11)--(r14)
		(b11)--(r15)
		(b21)--(r11)
		(b21)--(r12)
		(b21)--(r13)
		(b21)--(r14)
		(b21)--(r15)
		(b21)--(r21)
		(b21)--(r22)
		(b21)--(r23)
		(b21)--(r24)
		(b21)--(r25)
		(b31)--(r21)
		(b31)--(r22)
		(b31)--(r23)
		(b31)--(r24)
		(b31)--(r25)
		(b12)--(r11)
		(b12)--(r12)
		(b12)--(r13)
		(b12)--(r14)
		(b12)--(r15)
		(b22)--(r11)
		(b22)--(r12)
		(b22)--(r13)
		(b22)--(r14)
		(b22)--(r15)
		(b22)--(r21)
		(b22)--(r22)
		(b22)--(r23)
		(b22)--(r24)
		(b22)--(r25)
		(b32)--(r21)
		(b32)--(r22)
		(b32)--(r23)
		(b32)--(r24)
		(b32)--(r25);
\end{tikzpicture} 	\caption{Two black nodes connected by the gadget $E$ from \refthm{AdditiveApproxHardness} with $m=3$ and $o=5$.}
	\label{fig:ApproxAdditive}
\end{figure*}

\begin{proof}[Proof of \refthm{AdditiveApproxHardness}]
	
	We recall the proof of \refthm{NPhard}.
	If the given 3-regular graph $G$ of order $n$ has a Hamiltonian cycle, then there is an alignment with an \mmc of 1 from $G$ to the $n$-nodes cycle $C_n$, so $\distp(G,C_n)\leq1$.
	Otherwise, any alignment has an \mmc of at least 3, so $\distp(G,C_n)\geq\boundp{3}$.
	We show that we can distinguish the two cases even if we can only compute $\distp$ up to a constant additive error $\varepsilon$.
	For this we modify $G$ and $C_n$ into colored graphs $G'$ and $C'_n$ (see \refdef{Color} in \refapx{approx}) so that $\cdistp(G',C'_n)+\varepsilon<\boundp{2m}$ for some $m$ if $G$ has a Hamiltonian cycle, and $\cdistp(G',C'_n)\geq\boundp{2m}$ otherwise.
	According to \reflem{ColorConversion} (see \refapx{approx}), we can transform $G'$ and $C'_n$ into the non-colored graphs $G''$ and $C''_n$ such that $\cdistp(G',C'_n)=\distp(G'',C''_n)$.
	Hence, we can decide \hamcycle using an approximation algorithm for \Distp with an additive error guarantee of $\varepsilon$.
	
	First, we define the colored gadget $E$ that is visualized in \reffig{ApproxAdditive} and used to replace edges in $G$ and $C_n$ as part of their modification.
	To construct $E$ we create two paths $p^1$ and $p^2$ with $m$ blue nodes each.
	Then, for $1\leq i<m$ we remove the path edges $p^1_ip^1_{i+1}$ and $p^2_ip^2_{i+1}$, add $o$ new red nodes to the graph, and connect each of them to $p^1_i$, $p^1_{i+1}$, $p^2_i$, and $p^2_{i+1}$.
	In total $E$ consists of $2m$ blue nodes and $(m-1)\cdot o$ red nodes.
	
	Now we modify the graphs $G$ and $C_n$.
	We assign the black color to all existing nodes and replace each edge $vw$ in $G$ and $C_n$ by a gadget $E$ connecting $v$ and $w$ as visualized in \reffig{ApproxAdditive}.
	Precisely, we remove the edge $vw$ and connect each blue node in $E$ to both $v$ and $w$.
	In addition, we add $n$ isolated gadgets to the graph $C_n$.
	The resulting graphs are denoted as $G'$ and $C'_n$ in the following.
	Note that both $G'$ and $C'_n$ contain exactly $3n$ gadgets.
	
	We choose $m$ so that $\boundp{2m}>\boundp{m}+\varepsilon$ and $o$ so that $\boundp{\lceil o/2\rceil}\geq\boundp{2m}$.
	Both $G'$ and $C'_n$ are well-defined because the function $\fboundp$ is strictly increasing according to \reflem{MMCDistinguish}.
	The order of $G'$ and $C'_n$ is polynomial in $n$ and the following claim holds for our construction.
\begin{claim}
		If $G$ has a Hamiltonian cycle then $\cdistp(G',C'_n)\leq\boundp{m}$ and otherwise $\cdistp(G',C'_n)\geq\boundp{2m}$.
		\begin{claimproof}
First we assume that $G$ has a Hamiltonian cycle.
			Then there is an alignment $\pi$ with an \mmc of 1 from $G$ to the $n$-nodes cycle $C_n$ as we showed in the proof of \refthm{NPhard}.
			We convert $\pi$ into an alignment $\pi'$ from $G'$ to $C'_n$ as follows.
			For each black node $v$ in $G'$ we set $v^{\pi'}\coloneqq v^\pi$ assuming that the original nodes of $G$ are not renamed within the modification to $G'$.
			Each gadget in $C'_n$ is mapped to a gadget in $G'$ as follows.
			If $\pi$ maps an edge $vw$ in $G$ to an edge $v^\pi w^\pi$ in $C_n$, then $\pi'$ maps the gadget that connects the nodes $v$ and $w$ in $G'$ to the gadget that connects the nodes $v^\pi$ and $w^\pi$ in $C'_n$ in the natural way.
			If $\pi$ maps an edge $vw$ in $G$ to a non-edge $v^\pi w^\pi$ in $C_n$, then $\pi'$ maps the gadget that connects the nodes $v$ and $w$ in $G'$ to an isolated gadget in $C'_n$ in the natural way.
			This is possible because the 3-regular graph $G$ has $n$ edges more than $C_n$ and there are $n$ isolated gadgets in $C'_n$.
			
			Each node in $G$ is adjacent to exactly one edge that is mapped to a non-edge.
			Hence, each black node in $G'$ has exactly $m$ blue mismatched neighbors.
			These blue nodes do not have any further mismatches themselves, so their \mc is 1.
			Using \reflem{MMCupperG} we conclude $\cdistp(G',C'_n)\leq\boundp{m}$.
			
			Now assume that $G$ has no Hamiltonian cycle and $\cdistp(G',C'_n)<\boundp{2m}$.
			Then there exists an alignment $\tau$ from $G'$ to $C'_n$ such that $\mu_p(G'^\tau-H')<\boundp{2m}$.
			First we show that $\tau$ has to preserve the structure of the gadgets.
			If $\tau$ maps two blue nodes $v$ and $w$ in $G$ that are connected to the same $o$ red nodes to two blue nodes $v^\tau$ and $w^\tau$ in $C'_n$ that are not connected to the same $o$ red nodes, then either $v$ or $w$ has a \mc of at least $\lceil o/2\rceil$, so $\mu_p(G'^\tau-H')\geq\boundp{\lceil o/2\rceil}\geq\boundp{2m}$ according to \reflem{MMClowerG}, which is a contradiction to our assumption.
			
			Let $\pi$ be the alignment from $G$ to $C_n$ such that $v^\pi=v^\tau$ for all nodes $v$ in $G$.
			Since every node in $G$ has degree 3 and every node in $C_n$ has degree 2, we know that one edge of each node in $G$ is mapped to a non-edge in $C_n$.
			This means for each node $v$ in $G'$ that the blue nodes of one gadget attached to $v$ are mapped to the blue nodes of a gadget in $C'_n$ that is not attached to $v^\pi$.
			Hence, every black node in $G'$ has at already $m$ $\tau$-mismatches.
			Since $G$ does not have a Hamiltonian cycle, there has to be one additional edge $vw$ in $G$ that is not mapped to an edge in $C_n$.
			This means that the gadget connecting $v$ and $w$ in $G'$ cannot be connected to both $v^\tau$ and $w^\tau$ in $C'_n$.
			Hence, either $v$ or $w$ has a $\tau$-\mc of at least $2m$, so $\mu_p(G'^\tau-H')\geq\boundp{2m}$ according to \reflem{MMClowerG}, which is a contradiction to our assumption.
		\end{claimproof}
	\end{claim}
	
	Transforming $G'$ and $C'_n$ into the non-colored graphs $G''$ and $C''_n$ as described in \reflem{ColorConversion} leads to the following results.
	The order of $G''$ and $C''_n$ is polynomial in $n=\abs{V(G)}$.
	If $G$ does not have a Hamiltonian cycle, then $\distp(G'',C'')=\cdistp(G',C')\geq\boundp{2m}$.
	If $G$ has a Hamiltonian cycle then $\distp(G'',C'')=\cdistp(G',C')\leq\boundp{m}$, so $\distp(G'',C'')+\varepsilon\leq\boundp{m}+\varepsilon<\boundp{2m}$ using the condition for the choice of $m$.
	All in all we have shown that we can decide the \NP-hard problem \hamcycle on 3-regular graphs in polynomial time if we have a polynomial time approximation algorithm for \Distp with constant additive error.
\end{proof}

\begin{proof}[Proof of \refthm{ApproxAlgorithm}]
	Let $G$ and $H$ be graphs such that $H$ has maximum degree $d$.
	We consider the algorithm $\mathcal{A}$ that simply returns $\norm{A_G}_p+\norm{A_H}_p$.
	First, we show that $\mathcal{A}$ is correct using the triangle inequality:
	\begin{align*}
		\distp(G,H)
		=&\;\min_\pi\norm{A_G^\pi-A_H}_p\\
		\leq&\;\min_\pi\norm{A_G^\pi}_p+\norm{-A_H}_p\\
		=&\;\norm{A_G}_p+\norm{A_H}_p\\
		=&\;\mathcal{A}(G,H).
	\end{align*}
	Now we show that $\mathcal{A}$ has constant additive error. Choosing $\tau=\argmin_\pi\norm{A_G^\pi-A_H}$ we get			
	\begin{align*}
		\mathcal{A}(G,H)
		=&\;\norm{A_G}_p+\norm{A_H}_p\\
		=&\;\norm{A_G^\tau}_p+\norm{A_H}_p\\
		=&\;\norm{A_G^\tau-A_H+A_H}_p+\norm{A_H}_p\\
		\leq&\;\norm{A_G^\tau-A_H}_p+2\norm{A_H}_p\\
		\overset{(*)}{\leq}&\;\norm{A_G^\tau-A_H}_p+2d\\
		=&\;\min_\pi\norm{A_G^\pi-A_H}_p+2d\\
		=&\;\distp(G,H)+2d.
	\end{align*}
	Equation $(*)$ follows from the argument in the proof of \reflem{MCG} applied to $A_H$ instead of $A^\pi_G-A_H$.

	The approximation algorithm $\mathcal{A'}$ compares the maximum degree of $G$ to $d$.
	If they are equal, $\mathcal{A'}$ checks whether $G$ and $H$ are isomorphic.
	This can be done in polynomial time since both graphs have bounded degree \cite{Luks}.
	In the case of isomorphy we know that $\distp(G,H)=0$, so $\mathcal{A'}$ returns 0.
	Otherwise, $\mathcal{A'}$ returns $\mathcal{A}(G,H)$.
	We compute the following upper bound of the multiplicative error for the latter case:
	\begin{align}
		\frac{\mathcal{A}(G,H)}{\distp(G,H)}
		&\leq\frac{\distp(G,H)+2d}{\distp(G,H)}\\
		&=1+\frac{2d}{\distp(G,H)}\\
		&\leq1+d.\label{eq:MultiplicativeApproxAlgorithm1}
	\end{align}
	\refeq{MultiplicativeApproxAlgorithm1} follows from \reflem{MMClowerG} since $\boundp{c}\geq1$ for $c\geq1$ and any alignment from $G$ to $H$ has at least one mismatch if the graphs have a different maximum degree or the isomorphism test failed.
\end{proof}

\subsection{Missing Proofs from \refsec{CutNorm}}
\label{apx:cut}
\begin{proof}[Proof of Theorem~\ref{theo:cut}]
	The proof is done by reduction from the \NP-hard \maxcut problem on
	unweighted graphs \cite{Intractability}.
	First, we recall how Alan and Naor \cite{AlonNaor06} construct a matrix $A$ for any graph $G$ so that $\norm{A}_\square=MAXCUT(G)$. Orient $G$ in an arbitrary manner, let $V(G)=\{v_1,v_2,\dots,v_n\}$ and $E(G)=\{e_1,e_2,\dots,e_m\}$.
	Then $A$ is the $2m\times n$ matrix defined as follows.
	For each $1\leq i\leq m$, if $e_i$ is oriented from $v_j$ to $v_k$, then $A_{2i-1,j}=A_{2i,k}=1$ and $A_{2i-1,k}=A_{2i,j}=-1$.
	The rest of the entries in $A$ are all 0.

	Next, we observe that the matrix
	\begin{align*}
		B=\begin{pmatrix}
			0 & A\\
			A^T & 0\\
		\end{pmatrix}
	\end{align*}
	has the property $\norm{B}_\square=2\norm{A}_\square=2\cdot MAXCUT(G)$ since the cut norm is invariant under transposition and the two submatrices $A$, $A^T$ have no common rows or columns in $B$.
	
	We interpret $B$ as the adjacency matrix of a signed graph $\Delta'$ and construct the two unsigned graphs $F'\coloneqq(V(\Delta),E_+(\Delta))$, $H'\coloneqq(V(\Delta),E_-(\Delta))$.
	Then $\mu_\square(F'-H')=\mu_\square(\Delta')=\norm{B}_\square=2\cdot MAXCUT(G)$.
	Next, we modify $F'$, $H'$ into the graphs $F$, $H$ by adding $i\cdot\left(\left\lceil\frac{n^2}{4}\right\rceil+n\right)$ leaves to node $v_i$ for $1\leq i\leq n$.
The reduction follows from the claim that $\dist_\square(F,H)=2\cdot MAXCUT(G)$, which we prove in the following.
	
	Let $\pi$ be an alignment that maps $v_i$ to $v_i$ for $1\leq i\leq n$ and each leaf in $F$ to a leaf in $H$ so that its adjacency is preserved; let $\Delta$ be the mismatch graph of $\pi$.
	Then $E(\Delta)=E(\Delta')$ and therefore $\mu_\square(F^\pi-H)=2\cdot MAXCUT(G)$.
	We conclude $\dist_\square(F,H)\leq 2\cdot MAXCUT(G)$.
	
	It remains to show that no alignment can lead to a lower mismatch norm.
	First, let $\sigma$ be any alignment that maps $v_i$ to $v_i$ for $1\leq i\leq n$; let $\Lambda$ be the mismatch graph of $\sigma$.
	Then $\Lambda[\{v_1,\dots,v_n\}]=\Delta'$ and we get $\mu_\square(\Lambda)\geq 2\cdot MAXCUT(G)$ from \reflem{CutNormInducedSubgraph}.
	
	Conversely, let $\rho$ be any alignment that maps $v_i$ to $v_j$ for $i\neq j$ and $i,j\leq n$; let $\Gamma$ be the mismatch graph of $\rho$.
	The number of leaves adjacent to $v_i$ in $F$ and to $v_j$ in $H$ differs at least by $\left \lceil{\frac{n^2}{4}}\right \rceil+n$.
	Without restriction we can assume there are at least $l\coloneqq\left \lceil{\frac{n^2}{4}}\right \rceil$ leaves $w_1,\dots,w_l$ adjacent to $v_i$ in $F$ that are mismatched by $\rho$.
	Let $S\coloneqq\{v_i^\rho, w_1^\rho,\dots,w_l^\rho\}$.
	Then $\Gamma[S]$ has exactly $l$ edges all of which have the same sign.
	It is easy to see that $\mu_\square(\Gamma[S])=2l$ which implies $\mu_\square(\Gamma)\geq 2l$ according to \reflem{CutNormInducedSubgraph}.
	We chose $l$ so that $2l\geq\frac{n^2}{2}\geq 2\cdot MAXCUT(G)$.
	After considering all alignments, we get $\dist_\square(F,G)\geq 2\cdot MAXCUT(G)$. This proves our claim and therefore concludes the reduction.
\end{proof}

\end{document}